\documentclass[journal]{IEEEtran}

\usepackage{times}
\usepackage{graphicx}
\usepackage{amsmath}
\usepackage{amssymb}
\usepackage{url}
\usepackage{booktabs}
\usepackage{multirow}
\usepackage{xcolor}
\usepackage[T1]{fontenc}
\usepackage{listings}
\usepackage{tabularx}
\usepackage{algorithm}
\usepackage{algorithmic}
\usepackage[normalem]{ulem}
\usepackage[colorlinks=true, urlcolor=blue, linkcolor=black, citecolor=black]{hyperref}
\usepackage[switch]{lineno}

\usepackage{amsmath}
\usepackage{amssymb}
\usepackage{mathtools}

\usepackage{cite}
\usepackage{balance}
\usepackage[most]{tcolorbox}

\usepackage{soul}       

\usepackage{amsthm}

\newtheorem{theorem}{Theorem}
\newtheorem{corollary}{Corollary}
\newtheorem{proposition}{Proposition}

\usepackage{amsthm}

\makeatletter
\renewenvironment{IEEEproof}[1][Proof]{%
  \par\pushQED{\qed}%
  \normalfont
  \topsep6\p@\@plus6\p@\relax
  \trivlist
  \item[\hskip\labelsep
        \bfseries #1\@addpunct{.}]%
}{%
  \popQED\endtrivlist\@endpefalse
}
\makeatother

\newtheoremstyle{boldstyle}
  {3pt}   
  {3pt}   
  {\itshape}  
  {}      
  {\bfseries} 
  {.}     
  { }     
  {}      

\theoremstyle{boldstyle}

\newtcolorbox{takeawaybox}[1][]{%
  colback=gray!5,
  colframe=gray!60,
  coltitle=black,
  fonttitle=\bfseries,
  boxrule=0.4pt,
  arc=1pt,
  left=4pt,
  right=4pt,
  top=4pt,
  bottom=4pt,
  title=#1
}

\newcommand{\dagG}{\mathcal{G}}

\makeatletter

\def\@seccntformat#1{\csname the#1\endcsname\quad}

\makeatother

\title{Explainable Autonomous Cyber Defense using Adversarial Multi-Agent Reinforcement Learning}

\author{
Yiyao Zhang\textsuperscript{1},
Diksha Goel\textsuperscript{2},
Hussain Ahmad\textsuperscript{1}\\[6pt]
\textsuperscript{1}Adelaide University, Australia\\
\textsuperscript{2}CSIRO's Data61, Australia\\[4pt]
Email: yiyao.zhang@student.adelaide.edu.au; 
diksha.goel@csiro.au; hussain.ahmad@adelaide.edu.au\\
Corresponding author: Hussain Ahmad
}

\begin{document}


\maketitle

\begin{abstract}
Autonomous agents are increasingly deployed in both offensive and defensive cyber operations, creating high-speed, closed-loop interactions in critical infrastructure environments. Advanced Persistent Threat (APT) actors exploit "Living off the Land" techniques and targeted telemetry perturbations to induce ambiguity in monitoring systems, causing automated defenses to overreact or misclassify benign behavior as malicious activity. Existing monolithic and multi-agent defense pipelines largely operate on correlation-based signals, lack structural constraints on response actions, and are vulnerable to reasoning drift under ambiguous or adversarial inputs.

We present the Causal Multi-Agent Decision Framework (C-MADF), a structurally constrained architecture for autonomous cyber defense that integrates causal modeling with adversarial dual-policy control. C-MADF first learns a Structural Causal Model (SCM) from historical telemetry and compiles it into an investigation-level Directed Acyclic Graph (DAG) that defines admissible response transitions. This roadmap is formalized as a Markov Decision Process (MDP) whose action space is explicitly restricted to causally consistent transitions. Decision-making within this constrained space is performed by a dual-agent reinforcement learning system in which a threat-optimizing Blue-Team policy is counterbalanced by a conservatively shaped Red-Team policy. Inter-policy disagreement is quantified through a Policy Divergence Score and exposed via a human-in-the-loop interface equipped with an Explainability–Transparency Score (ETS) that serves as an escalation signal under uncertainty.

On the real-world CICIoT2023 dataset, C-MADF reduces the false-positive rate from 11.2 percent, 9.7 percent, and 8.4 percent in three cutting-edge literature baselines to 1.8 percent, while achieving 0.997 precision, 0.961 recall, and 0.979 F1-score. The ETS signal exhibits strong monotonic alignment with evidentiary sufficiency and policy agreement, supporting calibrated escalation of autonomous actions under uncertainty.
\end{abstract}

\begin{IEEEkeywords}
Autonomous Cyber Defense, Explainable AI, Causal Modeling, Multi-Agent Reinforcement Learning, Decision-Making, Adversarial Robustness.
\end{IEEEkeywords}

\section{Introduction}
\label{sec:introduction}

\IEEEPARstart{T}{he} rapid digital transformation of critical infrastructure has fundamentally altered the security and dependability landscape of cyber-physical systems \cite{goel2023enhancing, goel2018overview}. Modern environments increasingly incorporate microservice-based architectures \cite{ahmad2024smart, ahmad2025towards, ahmad2025resilient} and advanced machine learning components \cite{zhang2025regimefolio, chen20253s, jois2026australian}, while domains such as power grids and industrial control systems are now deeply interconnected with enterprise IT, cloud platforms, and remote management interfaces \cite{scada_security, tdsc_apt}. This integration improves efficiency but also introduces tightly coupled dependencies that amplify cyber risk and potential cascading failures under attack \cite{scada_security, ahmad2023review}.

Advanced Persistent Threat (APT) actors increasingly exploit this complexity using stealth-oriented strategies such as Living-off-the-Land behavior \cite{jayalath2024microservice}, telemetry manipulation, and staged lateral movement \cite{lotl_report, apt_survey}. In parallel, organizations are adopting autonomous and semi-autonomous defense pipelines, including reinforcement learning and LLM-assisted agents \cite{chopra2026chatnvd, abdulsatar2025towards}, to triage alerts and respond at machine speed \cite{rl_cyber, drl_attack, aica_book, goel2024optimizing}. Recent surveys on LLM-based autonomous cyber agents and hallucination behavior further indicate that unconstrained agent reasoning can amplify operational risk under uncertainty \cite{llm_agents_survey, hallucination}. While this shift improves responsiveness, it also introduces a key dependability problem: under ambiguous or adversarial observations, correlation-only automation can trigger unsafe mitigation actions, false positives, and operational disruption \cite{xai_survey, adversarial_ml, macas2023adversarial}.

Existing lines of work provide important but partial capabilities. Causal analytics improve post hoc reasoning but often remain detached from executable policy constraints \cite{tdsc_causal, causal_inference, abbas2024robust}; explainable AI improves interpretability but does not enforce admissible action trajectories \cite{xai_cyber, xai_survey, abbas2025scalar}; and multi-agent security systems improve decomposition yet frequently lack verifiable workflow restrictions \cite{autogen, tdsc_mas_security}. As a result, current approaches rarely provide an end-to-end mechanism that jointly delivers causal grounding, constrained decision execution, adversarial internal validation, and calibrated human escalation \cite{goel2025co, goel2025unveiling}.

To address this gap, we propose the Causal Multi-Agent Decision Framework (C-MADF), a structured architecture for dependable autonomous cyber defense. Operationally, C-MADF proceeds in four dependent stages: (i) telemetry is used to learn a Structural Causal Model (SCM), (ii) the learned SCM is compiled into a constrained Markov Decision Process (MDP)--Directed Acyclic Graph (DAG) investigation roadmap, (iii) Blue-Team and Red-Team policies deliberate over admissible actions, and (iv) a human-in-the-loop interface computes ETS to govern autonomous execution versus escalation.\\

The main contributions of this paper are as follows:
\begin{itemize}
    \item \textbf{Causally constrained investigation planning.} We learn a SCM from security telemetry and compile it into an investigation-level MDP-DAG that restricts autonomous actions to causally admissible transitions.
    \item \textbf{Adversarial dual-policy deliberation.} We introduce a Council of Rivals design in which a threat-optimizing Blue-Team policy is counterbalanced by a conservative Red-Team policy, with disagreement quantified as a Policy Divergence Score for uncertainty-aware arbitration.
    \item \textbf{Explainable escalation mechanism.} We define an Explainability--Transparency Score (ETS) that aggregates explanation quality, evidentiary sufficiency, and policy consistency to support transparent human-in-the-loop intervention.
    \item \textbf{Verification-oriented evaluation.} We provide formal and empirical analysis, including constrained-transition reasoning, baseline comparisons, and calibration checks on the real-world CICIoT2023 dataset.
\end{itemize}

The remainder of this paper is organized as follows.
Section~\ref{sec:related_work} reviews related work and clarifies the research gap.
Section~\ref{sec:preliminaries}  introduces the formal preliminaries, and Section~\ref{sec:problem_formulation} defines the system and adversary models.
Section~\ref{sec:cmadf_framework} details the proposed framework.
Section~\ref{sec:security_analysis} presents security and verification-oriented analysis.
Section~\ref{sec:evaluation} reports quantitative results on CICIoT2023.
Section~\ref{sec:discussion} discusses implications and limitations, and Section~\ref{sec:conclusion} concludes the paper. Besides, Table~\ref{tab:acronyms} shows the definitions used in the paper.

\color{black}

\begin{table*}[t]
\centering
\small
\caption{Acronyms and standardized terminology used in this paper}
\label{tab:acronyms}
\begin{tabularx}{\textwidth}{@{} l X @{}}
\toprule
\textbf{Acronym/Term} & \textbf{Definition used in this paper} \\
\midrule
C-MADF & Causal Multi-Agent Decision Framework \\
APT & Advanced Persistent Threat \\
SCM & Structural Causal Model \\
DAG & Directed Acyclic Graph \\
MDP & Markov Decision Process \\
MDP-DAG roadmap & Causally constrained investigation roadmap compiled from SCM into executable MDP transitions \\
ETS & Explainability--Transparency Score \\
AICA & Autonomous Intelligent Cyber Defense Agent \\
XAI & Explainable Artificial Intelligence \\
LLM & Large Language Model \\
MADRL & Multi-Agent Deep Reinforcement Learning \\
ICS & Industrial Control System \\
IT/OT & Information Technology / Operational Technology \\
IoT & Internet of Things \\
CICIoT2023 & Canadian Institute for Cybersecurity IoT Dataset (2023) \\
\bottomrule
\end{tabularx}
\end{table*}

\color{black}

\section{Background and Related Work}
\label{sec:related_work}

Autonomous cyber defense has progressed significantly in response to increasingly adaptive adversaries and the operational complexity of modern cyber-physical infrastructure. In this section, we review the principal research directions that inform C-MADF and identify structural limitations that motivate our framework. Specifically, we examine prior work in: (i) causal inference for security analytics, (ii) machine learning for intrusion detection, (iii) autonomous cyber-defense agents, (iv) explainable AI in cybersecurity, (v) multi-agent security architectures, and (vi) adversarial machine learning. The comparison results can be found in Table \ref{tab:related_work_comparison}. Related strands on immersive cyber situational awareness, game-theoretic defense models, and sector-specific IoT security assessments further reinforce the need for structured and operator-aware cyber defense workflows \cite{ahmad2025survey, game_theory_security, security}. 

These research directions highlight both the strengths and limitations of current approaches in handling complex and adversarial environments.

Our analysis emphasizes a common limitation across these lines of work: although individually powerful, existing approaches do not simultaneously provide causally grounded investigation logic, formally constrained decision trajectories, and adversarial multi-agent deliberation within a verifiable decision-theoretic framework.

\begin{table*}[t]
\centering
\small
\setlength{\tabcolsep}{5pt}  

\caption{Comparison of C-MADF with Representative Autonomous Cyber Defense Approaches}
\label{tab:related_work_comparison}
\begin{tabular}{p{4.5cm} c c c c c}
\toprule
\textbf{Approach} &
\textbf{Formally Structured} &
\textbf{Causal} &
\textbf{Multi-Agent} &
\textbf{Integrated} &
\textbf{Decision-Level} \\
 &
\textbf{Decision Model} &
\textbf{Grounding} &
\textbf{Deliberation} &
\textbf{XAI} &
\textbf{Robustness} \\
\midrule

Causal Security Analytics \cite{tdsc_causal, causal_inference} &
$\triangle$ (offline) & $\checkmark$ & $\times$ & $\triangle$ & $\triangle$ \\

ML for Intrusion Detection \cite{sujamary2024network, imrana2021bidirectional, devendiran2024dugat} &
$\times$ & $\times$ & $\times$ & $\times$ & $\triangle$ (data-level) \\

AICA-style Agents \cite{aica_book} &
$\triangle$ & $\times$ & $\times$ & $\times$ & $\triangle$ \\

XAI for Cybersecurity \cite{xai_cyber, xai_survey} &
$\times$ & $\times$ & $\times$ & $\checkmark$ (post-hoc) & $\times$ \\

MAS in Cybersecurity \cite{autogen, tdsc_mas_security} &
$\times$ & $\times$ & $\checkmark$ & $\times$ & $\triangle$ \\

Adversarial ML Defenses \cite{adversarial_ml} &
$\times$ & $\times$ & $\times$ & $\times$ & $\checkmark$ (model-level) \\

\midrule
\textbf{C-MADF (this work)} &
\textbf{$\checkmark$ (MDP-DAG)} &
\textbf{$\checkmark$ (online SCM)} &
\textbf{$\checkmark$ (Council)} &
\textbf{$\checkmark$ (ETS)} &
\textbf{$\checkmark$ (process-level)} \\

\bottomrule
\end{tabular}

\vspace{0.3em}
\footnotesize \textit{Legend:}
$\checkmark$ = explicitly supported;
$\triangle$ = partial or context-dependent support;
$\times$ = not explicitly addressed.
\end{table*}

\subsection{Causal Inference in Security}

Causal inference has been increasingly applied to cybersecurity tasks including root cause analysis, anomaly detection, and attack attribution \cite{tdsc_causal, causal_inference, spirtes2000causation}. By modeling cause–effect relationships rather than correlations, causal approaches improve robustness to spurious associations and distributional shifts. While causal models have been applied to enhance anomaly detection and improve data quality in security analytics \cite{gayathri2024hybrid}, these applications primarily focus on improving the fidelity of training data and detection accuracy rather than constraining autonomous decision processes. For instance, generative models combined with causal reasoning have been used to address data imbalance in insider threat detection, but the learned causal relationships are not compiled into enforceable workflow constraints for autonomous response systems.

Nonetheless, most existing work applies causal models for offline analysis or enhanced detection, rather than for constraining online decision processes. Causal security analytics typically stop at producing improved explanatory or predictive models; they do not govern which autonomous actions are admissible in real time.

C-MADF extends this line of work by compiling a learned SCM into an MDP-DAG roadmap that explicitly restricts allowable state transitions. This integration enables causal consistency to function as a hard constraint on action sequences, rather than solely as an explanatory overlay.

\subsection{Machine Learning for Intrusion Detection}

Machine learning-based intrusion detection has evolved significantly \cite{ahmad2025future, ullah2026skills}, with contemporary approaches achieving high accuracy through ensemble methods \cite{mohiuddin2023intrusion}, deep neural architectures \cite{sujamary2024network, imrana2021bidirectional}, and bio-inspired optimization \cite{alazab2022new, devendiran2024dugat}. Bidirectional LSTM networks have proven particularly effective for sequential threat pattern recognition \cite{imrana2021bidirectional}, while hybrid meta-heuristic techniques combined with gradient boosting classifiers have demonstrated robustness across diverse attack scenarios \cite{mohiuddin2023intrusion}.

These advances are underpinned by foundational deep-learning principles for representation learning and optimization \cite{goodfellow2016deep}.

Recent systematic reviews highlight the expansion of ML-based intrusion detection to IoT environments \cite{goel2024machine}, where resource constraints and heterogeneous device characteristics introduce additional challenges \cite{essa2024review}. However, several fundamental limitations persist. First, model reliability can be compromised by data preprocessing artifacts and pattern leakage, particularly when training and deployment distributions diverge \cite{bouke2023empirical}. Second, optimization strategies for improving detection accuracy, such as chaotic optimization in deep learning architectures \cite{devendiran2024dugat}, do not address the interpretability requirements for autonomous response systems. Third, decision tree-based ensemble methods, while offering improved classification performance \cite{sun2024improved}, lack the causal grounding necessary to support verifiable investigation workflows.

These approaches excel at threat detection but do not provide mechanisms for causally constrained autonomous response, adversarial validation of mitigation actions, or structured human oversight under epistemic uncertainty. C-MADF complements detection-focused ML systems by providing a principled framework for translating detection outputs into verifiable, explainable autonomous actions.

\subsection{Autonomous Intelligent Cyber-defense Agents (AICA)}

Autonomous Intelligent Cyber-defense Agents (AICA) have been proposed as scalable mechanisms for continuous monitoring and automated response in contested environments \cite{aica_book}. Despite their promise, existing AICA implementations exhibit three critical structural limitations. \textbf{First}, decision-making logic is predominantly correlation-based rather than causally grounded, remaining vulnerable to LotL techniques \cite{lotl_report} and zero-day exploits \cite{zero_day_detection}. Under adversarial telemetry manipulation, correlation-based systems cannot distinguish genuine causal dependencies from spurious co-occurrence \cite{causal_inference}. \textbf{Second}, internal decision processes lack principled auditability; while recent proposals incorporate post-hoc explanation modules \cite{xai_survey}, these operate as interpretive overlays rather than structural constraints on action admissibility. \textbf{Third}, AICA frameworks lack formally constrained investigation pathways, potentially leading to self-inflicted denial-of-service events under adversarial ambiguity \cite{tdsc_apt}.

High detection accuracy \cite{mohiuddin2023intrusion, sujamary2024network, imrana2021bidirectional} and autonomous response reliability are orthogonal properties \cite{bouke2023empirical}. AICA-style agents do not embed policies within learned SCMs or constrain mitigations through decision-theoretic roadmaps amenable to formal analysis, precluding trajectory-level guarantees essential for critical infrastructure certification \cite{tdsc_formal_verification}.

C-MADF addresses these gaps by integrating (i) a learned SCM that filters spurious correlations, (ii) a causally constrained MDP-DAG restricting policy exploration to valid investigation trajectories, and (iii) adversarial multi-agent deliberation exposing epistemic uncertainty through calibrated disagreement.

\subsection{Explainable AI (XAI) in Cybersecurity}

As AI systems assume increasingly critical roles in security operations, explainability has become essential for trust, accountability, and safe deployment \cite{xai_cyber, xai_survey}. Foundational work on explanation in AI provides an additional conceptual basis for this requirement \cite{tdsc_xai}. Existing XAI techniques include feature attribution (e.g., Local Interpretable Model-agnostic Explanations (LIME), SHapley Additive exPlanations (SHAP)), rule extraction, and counterfactual explanations.

However, most cybersecurity-focused XAI approaches are post-hoc and model-centric. They explain static predictions rather than constraining the evolution of autonomous decision policies over time. In interactive cyber defense environments, explanations must be generated in real time and influence whether actions are executed.

While deep learning models have achieved state-of-the-art performance in network intrusion detection tasks \cite{sujamary2024network, imrana2021bidirectional, devendiran2024dugat}, their black-box nature poses significant challenges for operational deployment in security-critical environments. Recent work has explored various explainability techniques tailored to cybersecurity applications, yet these approaches typically provide post-hoc rationalizations without influencing the admissibility of autonomous actions \cite{essa2024review}. The integration of explainability as a structural constraint on decision processes, rather than as an interpretive overlay, remains an open challenge.

Recent surveys on counterfactual explanations and algorithmic recourse further highlight the need for explanations that support actionable intervention pathways \cite{causal_xai}.

C-MADF integrates explainability as a structural property of the decision process rather than as a wrapper around a trained model. The ETS is embedded within the multi-agent reinforcement learning objective and is directly tied to the constrained MDP-DAG structure. Thus, interpretability signals influence policy learning and execution gating, rather than merely post-hoc justification.

This design choice is also aligned with broader human-in-the-loop machine learning principles for safety, accountability, and calibrated oversight in high-stakes domains \cite{human_in_loop_ml}.

\subsection{Multi-Agent Systems in Cybersecurity}

Multi-agent systems have been introduced to decompose complex cybersecurity tasks across specialized agents \cite{autogen, tdsc_mas_security}. Collaborative monitoring, distributed reasoning, and coordinated response mechanisms have demonstrated improved coverage and adaptability compared to monolithic architectures.

Foundational multi-agent formulations and modern MADRL critiques both emphasize coordination--stability trade-offs that become critical under adversarial uncertainty \cite{shoham2008multiagent, madrl_survey}.

However, existing MAS approaches typically operate within unconstrained state spaces. Under ambiguous or adversarial evidence, agents may converge toward mutually reinforcing but incorrect threat narratives, a phenomenon akin to reasoning drift. Because decision trajectories are not structurally constrained by formal workflow models, these systems lack guarantees that disruptive actions are preceded by appropriate evidence-gathering states.

Moreover, representative MAS architectures do not ground their action spaces in explicitly learned causal models, nor do they embed agent deliberation within a decision process suitable for formal verification. Consequently, prior MAS work does not provide enforceable properties such as "no mitigation action without traversing required investigative states," which is a central objective of C-MADF.

Recent work has explored adversarial robustness in distributed security systems, demonstrating that deep learning-enabled multi-agent architectures remain vulnerable to carefully crafted adversarial examples that can evade detection or provoke unintended defensive responses \cite{macas2023adversarial}. While these studies highlight the importance of adversarial training and input sanitization at the classifier level, they do not address the structural constraints necessary to prevent unsafe multi-step mitigation sequences or provide formal guarantees over investigation trajectories, gaps that C-MADF directly addresses through its causally constrained MDP-DAG architecture.

\subsection{Adversarial Machine Learning}

Adversarial machine learning has systematically demonstrated that deep learning models deployed in cybersecurity applications are vulnerable to carefully crafted input perturbations that induce confident yet incorrect predictions \cite{adversarial_ml}. In security settings, adversarial examples can enable evasion attacks, trigger false positives in intrusion detection systems, or provoke unintended defensive actions across domains such as network intrusion detection, malware classification, and insider threat analysis.

Consistent with Table 1, the adversarial-ML studies considered here focus on defenses such as adversarial training, input sanitization, and augmentation-based robustness mechanisms at the classifier or perception level; however, these methods do not reason over long-horizon investigation workflows, impose structural constraints on mitigation sequences, or provide trajectory-level guarantees that disruptive actions are preceded by required evidence-gathering states. \cite{adversarial_ml}

C-MADF complements model-level robustness techniques by introducing structured counterfactual scrutiny at the decision-process level. A conservative Red-Team policy challenges proposed Blue-Team interventions within a causally constrained state space, thereby reducing unsafe mitigation trajectories under epistemic uncertainty.

\subsection{Summary and Research Gap}
\label{sec:gap_analysis}

\begingroup

However, several core ingredients of C-MADF are now established in isolation, and this manuscript must be more explicitly positioned against the latest state-of-the-art work in closely related areas.

The prior literature provides essential components for autonomous cyber defense but does not unify them within a single verifiable framework:

\begin{itemize}
    \item Causal analytics \cite{tdsc_causal, causal_inference} improve root cause analysis but do not constrain online decision processes.
    \item ML-based intrusion detection systems \cite{sujamary2024network, mohiuddin2023intrusion, imrana2021bidirectional, devendiran2024dugat, alazab2022new} achieve high detection accuracy but focus on classification rather than autonomous investigation and response, and remain vulnerable to data quality issues and distributional shift \cite{bouke2023empirical}.
    \item AICA systems \cite{aica_book} emphasize autonomy but lack causal grounding and formally constrained investigation paths.
    \item XAI approaches \cite{xai_cyber, xai_survey} enhance interpretability but do not govern real-time policy admissibility.
    \item Security-oriented MAS \cite{autogen, tdsc_mas_security} improve task decomposition yet operate in unconstrained state spaces without enforceable trajectory-level guarantees.
    \item Adversarial ML defenses \cite{adversarial_ml} improve model robustness but do not constrain long-horizon decision sequences.
\end{itemize}

C-MADF addresses this gap by integrating (i) learned SCM, (ii) an MDP-DAG decision process that enforces admissible action sequences, and (iii) adversarial multi-agent deliberation coupled with quantitative explainability signals. This combination yields a causally grounded, verifiable, and explainable multi-agent controller for autonomous cyber defense.

This gap is not merely terminological; it is operational. In practice, systems that optimize only detector-level accuracy can still produce unsafe intervention trajectories when evidence is sparse, contradictory, or strategically manipulated. By explicitly separating evidence acquisition states from high-impact mitigation states and enforcing admissibility through the roadmap, C-MADF targets this deployment-critical failure mode directly rather than indirectly via threshold tuning.

Accordingly, our contribution is positioned as a systems-level integration advance: we do not claim novelty for each individual building block in isolation, but for the end-to-end coupling of causal structure learning, constrained sequential control, adversarial policy cross-examination, and escalation-aware explainability under one coherent decision pipeline.
\endgroup

\section{Preliminaries}
\label{sec:preliminaries}

This section formalizes the mathematical foundations underlying C-MADF. We introduce the MDP formulation used to model constrained investigation workflows and the SCM framework used to encode causal dependencies among security-relevant variables. Notation introduced here is used consistently throughout the paper. The reference table can be found in Table \ref{tab:notation}.

\subsection{MDP}

We model the autonomous investigation and mitigation workflow as a finite-horizon discounted MDP,
\[
\mathcal{M} = \langle \mathcal{S}, \mathcal{A}, P, R, \gamma \rangle.
\]
Our notation follows canonical treatments in MDP and reinforcement learning literature \cite{mdp_book, sutton_barto}.
Table \ref{tab:acronyms} presents the symbol-level definitions used in the paper.

\medskip

\noindent\textbf{Constrained Action Space.}
In C-MADF, the action space $\mathcal{A}$ is not unconstrained. Instead, admissible actions at state $s$ are restricted to a subset $\mathcal{A}_C(s) \subseteq \mathcal{A}$ determined by a learned causal investigation graph (defined below). This structural restriction differentiates our formulation from conventional unconstrained MDP-based cyber-defense models.

\subsection{SCM}

To encode causal dependencies among investigation-relevant variables, we adopt the SCM framework \cite{causal_inference}.

An SCM is defined as a tuple
\[
\mathcal{M}_C = \langle U, V, F \rangle,
\]
where:

\begin{itemize}
    \item $U$ is a set of exogenous (unobserved) variables.
    \item $V = \{V_1, \dots, V_n\}$ is a set of endogenous variables.
    \item $F = \{f_1, \dots, f_n\}$ is a collection of structural equations such that
    \[
    V_i = f_i(\mathrm{Pa}_i, U_i),
    \]
    where $\mathrm{Pa}_i \subseteq V \setminus \{V_i\}$ denotes the set of parents (direct causes) of $V_i$ in the model.
\end{itemize}

The SCM induces a DAG $\mathcal{G} = (V, E)$,

\subsection{Assumptions}

Throughout the paper, we make the following assumptions:

\begin{itemize}
    \item The learned SCM is acyclic and provides a probabilistic structural approximation of the true (unknown) causal dependencies among investigation-relevant variables. We do not assume perfect recovery of the ground-truth causal graph. Instead, we assume bounded structural mis-specification, characterized by an edge-level error rate $\varepsilon$ relative to the true causal graph.
    
    \item The transition dynamics of the investigation process satisfy the Markov property at the level of the abstract roadmap states, i.e., transition probabilities are conditionally independent of prior history given the current state and action.
    
    \item Structural errors in the learned SCM are bounded such that their induced effect on admissible roadmap transitions is limited to a controlled fraction of state–action pairs. All probabilistic safety and bounded-violation results presented later are conditional on this bounded structural error assumption.
\end{itemize}

These assumptions reflect common modeling practices in causal discovery and reinforcement learning. However, we emphasize that our safety and robustness guarantees are conditional on bounded causal mis-specification rather than on exact structural recovery.

\begin{table}[t]
\centering
\caption{Summary of Core Notation}
\label{tab:notation}
\begin{tabularx}{\columnwidth}{@{} l X @{}}
\toprule
\textbf{Symbol} & \textbf{Description} \\ 
\midrule

$\mathcal{S}, \mathcal{A}$ & State space and action space of the MDP. \\

$P(s' \mid s,a), R(s,a), \gamma$ & Transition kernel, reward function, and discount factor. \\

$\pi, V^{\pi}, Q^{\pi}$ & Policy, state-value function, and action-value function under policy $\pi$. \\

$\mathcal{G} = (\mathcal{V}, E)$ & DAG induced by the learned SCM. \\

$\mathcal{M}_C = \langle U, \mathcal{V}, F \rangle$ & SCM with exogenous variables $U$, endogenous variables $\mathcal{V}$, and structural equations $F$. \\

$\pi_B, \pi_R$ & Blue-Team and Red-Team policies in the dual-agent architecture. \\

$\theta_B, \theta_R$ & Parameters of the Blue-Team and Red-Team policy networks. \\

$\mathcal{D}(s_t)$ & Policy Divergence Score measuring disagreement between $\pi_B$ and $\pi_R$ at state $s_t$. \\

$\mathcal{A}_C(s)$ & Causally admissible action subset at state $s$ induced by $\mathcal{G}$. \\

ETS & Explainability–Transparency Score used for execution gating and escalation. \\

\bottomrule
\end{tabularx}
\end{table}

\section{Problem Formulation}
\label{sec:problem_formulation}

We formalize the problem of causally grounded, verifiable, and explainable autonomous cyber defense under adversarial telemetry manipulation. We specify (i) the system model, (ii) the adversary model, and (iii) the design objectives that C-MADF must satisfy.

\subsection{System Model}

We consider a cyber-physical network consisting of a set of nodes 
\[
\mathcal{N} = \{n_1, \dots, n_m\},
\]
representing servers, workstations, industrial controllers, and networking devices. The global system state evolves over discrete time steps and is partially observed through telemetry streams 
\[
\mathcal{T}_t = \{ \text{logs}, \text{flows}, \text{system events}, \text{performance metrics} \}.
\]

A Security Operations Center (SOC) monitors $\mathcal{T}_t$ and deploys an autonomous defense stack composed of interacting agents. These agents may execute actions from a predefined action set 
\[
\mathcal{A} = \{\text{collect evidence}, \text{isolate host}, \text{block flow}, \dots\},
\]
subject to policy constraints defined in Section~\ref{sec:preliminaries}.

Agents share a learned SCM and operate within a causally constrained MDP. Action execution may be automated or escalated to a supervisory adjudication. The supervisory adjudication serves as the final authority for high-impact interventions and has access to structured explanations and uncertainty signals.

\subsection{Adversary Model}

We consider a strong but bounded adversary consistent with Advanced Persistent Threat (APT) capabilities. The adversary may exploit several vectors to disrupt the system, which are systematically summarized and formally described in Table~\ref{tab:threat_model}.

Capability-level definitions are presented in the Table~\ref{tab:threat_model}.

\subsubsection{Constraints.}
\label{subsubsec:adversary_constraints}

Scenario constraints under the partial-compromise assumption are summarized in a dedicated table: Table~\ref{tab:adversary_constraints}.

\begin{table}[t]
\centering
\caption{Adversary Capabilities under the Partial-Compromise Model}
\label{tab:threat_model}
\begin{tabularx}{\columnwidth}{@{} >{\hsize=0.8\hsize}X >{\hsize=1.2\hsize}X @{}}
\toprule
\textbf{Capability} & \textbf{Formal Description} \\ 
\midrule

Initial Compromise &
Obtains unauthorized access to a subset $\mathcal{N}_c \subset \mathcal{N}$ via credential theft, phishing, or exploitation of vulnerabilities. \\

Lateral Movement &
Moves within $\mathcal{N}$ using legitimate administrative tools and credentials (e.g., LotL techniques) to minimize detection signals. \\

Selective Telemetry Manipulation &
Drops, delays, replays, injects, or perturbs subsets of host-, network-, or control-layer telemetry streams to induce ambiguity in state estimation. \\

Partial Historical Poisoning &
Introduces bounded perturbations into a fraction of historical logs used for causal discovery, without full rewriting of all historical data. \\

Evasion Strategies &
Mimics benign operational patterns and exploits monitoring blind spots to delay detection or trigger false mitigation decisions. \\

\bottomrule
\end{tabularx}
\end{table}

\begin{table}[t]
\centering
\caption{Scenario Constraints under the Partial-Compromise Model}
\label{tab:adversary_constraints}
\begin{tabularx}{\columnwidth}{@{} >{\hsize=0.9\hsize}X >{\hsize=1.1\hsize}X @{}}
\toprule
\textbf{Constraint} & \textbf{Scenario Restriction} \\ 
\midrule

Node-Level Compromise Bound &
Adversary control is limited to a strict subset $\mathcal{N}_c \subset \mathcal{N}$; complete network takeover is excluded. \\

Telemetry Survivability &
The adversary cannot suppress all sensing channels simultaneously; at least one trustworthy telemetry path remains observable. \\

Defender Runtime Integrity &
The C-MADF execution environment, source code, and runtime parameters are not directly compromiseable by the adversary. \\

Secure-Channel and Crypto Integrity &
The adversary cannot subvert secure communication channels or break underlying cryptographic primitives. \\

Out-of-Scope Full-Compromise Attacks &
Full compromise of the defense stack, arbitrary code execution inside C-MADF, and complete rewriting of all historical data are out of scope. \\

\bottomrule
\end{tabularx}
\end{table}

\begin{figure*}[h!]
    \centering
    \includegraphics[width=\textwidth]{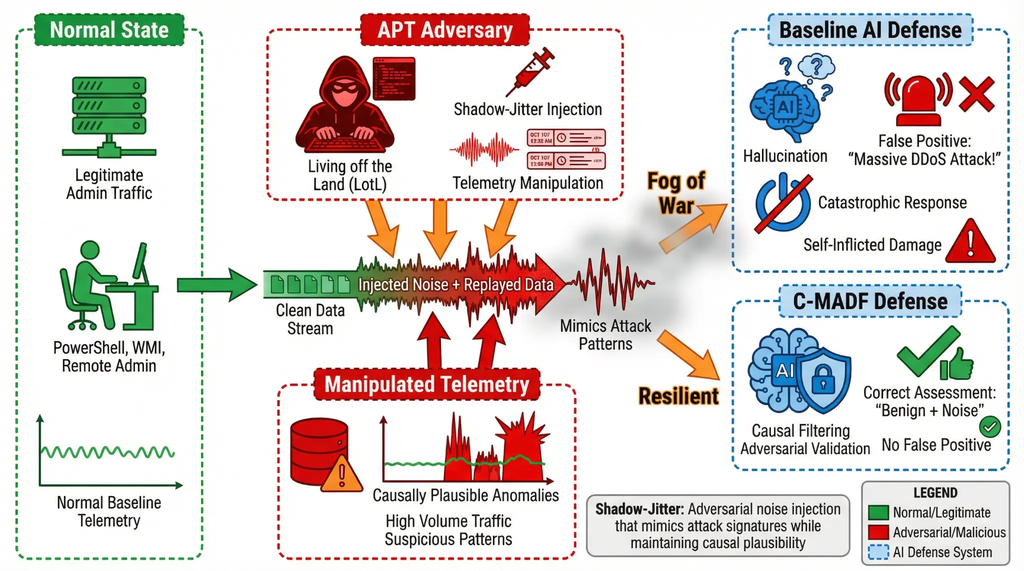}
\caption{Illustration of a Shadow-Jitter telemetry manipulation scenario. Controlled perturbations in host and network logs distort apparent event correlations, leading an unconstrained correlation-based defense to misclassify benign activity as malicious. In contrast, C-MADF applies causal filtering and adversarial dual-policy validation within a constrained MDP-DAG structure, reducing spurious mitigation actions under ambiguous observations.}
    \label{fig:shadow_jitter_attack}
\end{figure*}

\subsubsection{Formal Definition of Shadow-Jitter Telemetry Manipulation}

We now formalize the Shadow-Jitter adversary introduced informally in Section~1. 
Shadow-Jitter models a bounded, partially controlled telemetry perturbation 
process designed to distort apparent causal relationships while preserving 
superficial plausibility.\\

The conceptual mechanics of this adversary are illustrated in Figure~\ref{fig:shadow_jitter_attack}. 
The figure demonstrates how the injection of controlled noise into host and network logs distorts apparent event correlations. This visualization is critical for understanding how conventional unconstrained defenses are misled into perceiving a standard data exfiltration attack, whereas C-MADF uses this as a basis for causal filtering and adversarial validation.

\noindent \textit{Telemetry Representation.} Let the observed telemetry stream at discrete time $t$ be represented 
as a feature vector

\begin{equation}
T_t \in \mathbb{R}^d,
\end{equation}

where $d$ denotes the number of aggregated host-, network-, and 
control-level features extracted from raw logs and event streams 
(Section~\ref{sec:problem_formulation}).

Let

\begin{equation}
\mathbf{T}_{0:T} = \{T_0, T_1, \dots, T_T\}
\end{equation}

denote a telemetry trajectory over horizon $T$.

\noindent \textit{Bounded Perturbation Model.} A Shadow-Jitter adversary applies a perturbation operator 
$\mathcal{A}_{SJ}$ to the telemetry stream such that

\begin{equation}
\tilde{T}_t = \mathcal{A}_{SJ}(T_t) = T_t + \delta_t,
\end{equation}

where $\delta_t \in \mathbb{R}^d$ is an adversarial perturbation vector 
and $\tilde{T}_t$ denotes manipulated telemetry.

We impose a bounded perturbation constraint:

\begin{equation}
\|\delta_t\|_p \leq \epsilon, \quad \forall t,
\end{equation}

for some $p \in \{1,2,\infty\}$, where:

\begin{itemize}
    \item $\epsilon > 0$ is the attack strength parameter,
    \item $\|\cdot\|_p$ denotes the $\ell_p$-norm.
\end{itemize}

The parameter $\epsilon$ controls the magnitude of instantaneous 
telemetry distortion.\\

\noindent \textit{Partial Channel Control.} Consistent with the partial-compromise assumption 
(Section~\ref{subsubsec:adversary_constraints}), the adversary controls only a subset 
of telemetry channels $\mathcal{C} \subset \{1,\dots,d\}$ such that

\begin{equation}
\delta_t^i = 0 \quad \forall i \notin \mathcal{C},
\end{equation}

and

\begin{equation}
\frac{|\mathcal{C}|}{d} \leq \kappa,
\end{equation}

where $\kappa \in (0,1)$ denotes the compromised-channel fraction.

This models selective log suppression, injection, replay, or 
distortion at compromised nodes only.

\noindent \textit{Temporal Distortion (Replay/Delay).} To model timing-based perturbations, the adversary may introduce 
bounded temporal shifts:

\begin{equation}
\tilde{T}_t^i = T_{t - \Delta_i(t)}^i,
\end{equation}

subject to

\begin{equation}
0 \leq \Delta_i(t) \leq \Delta_{\max},
\end{equation}

where $\Delta_{\max}$ denotes the maximum allowable delay and 
distortion is restricted to channels $i \in \mathcal{C}$.

\noindent \textit{Partial Historical Poisoning.} Let $D$ denote the historical dataset used to learn the SCM. The adversary may construct

\begin{equation}
\tilde{D} = D \cup D_{\text{poison}},
\end{equation}

subject to

\begin{equation}
\frac{|D_{\text{poison}}|}{|D|} \leq \rho,
\end{equation}

where $\rho \in (0,1)$ denotes the poisoning rate.

The poisoning process must preserve marginal plausibility 
constraints to avoid trivial detection.

\noindent \textit{Shadow-Jitter Adversary.} We define the Shadow-Jitter adversary as the bounded parameter tuple

\begin{equation}
\mathcal{A}_{SJ} = (\epsilon, \kappa, \Delta_{\max}, \rho),
\end{equation}

such that telemetry manipulation satisfies:

\begin{equation}
\|\delta_t\|_p \leq \epsilon,
\quad
\frac{|\mathcal{C}|}{d} \leq \kappa,
\quad
\Delta_i(t) \leq \Delta_{\max},
\quad
\frac{|D_{\text{poison}}|}{|D|} \leq \rho.
\end{equation}

This adversary model captures:

\begin{itemize}
    \item bounded amplitude perturbation,
    \item partial channel compromise,
    \item bounded temporal distortion,
    \item limited historical poisoning.
\end{itemize}

All robustness and security guarantees in 
Section~\ref{sec:security_analysis} are defined with respect to 
this bounded adversarial model.

\subsection{Design Objectives}

Table~\ref{tab:goals_mapping} presents the objective definitions and corresponding C-MADF mechanisms.

\begin{table}[t]
\centering
\caption{Mapping Design Objectives to C-MADF Mechanisms}
\label{tab:goals_mapping}
\begin{tabularx}{\columnwidth}{@{} l 
    >{\raggedright\arraybackslash\hsize=0.8\hsize}X 
    >{\raggedright\arraybackslash\hsize=1.2\hsize}X @{}}
\toprule
\textbf{ID} & \textbf{Objective Definition} & \textbf{Mechanisms in C-MADF} \\ 
\midrule

G1 & Maintain high detection recall while limiting false-positive mitigation actions under bounded telemetry manipulation and partial historical poisoning. & Causally constrained action space derived from learned DAG $\mathcal{G}$; dual-policy adversarial deliberation; divergence-based escalation under high policy disagreement $\mathcal{D}(s_t)$. \\ \addlinespace

G2 & Ensure every executable mitigation action corresponds to an admissible transition in a formally defined investigation structure. & MDP-DAG structure restricting admissible state transitions; complete logging of state--action trajectories; compatibility with temporal-logic and model-checking analysis. \\ \addlinespace

G3 & Provide structured evidence traces and uncertainty signals for each proposed or executed action to support human oversight. & Human-in-the-loop interface exposing evidence traces and alternative hypotheses; ETS aggregating clarity, coverage, and policy uncertainty; explicit visualization of policy divergence $\mathcal{D}(s_t)$. \\ \addlinespace

G4 & Scale to large networks and high-volume telemetry streams without exponential growth in online decision complexity. & Modular architecture separating offline causal discovery from online decision-making; compact dual-policy networks $\pi_B, \pi_R$; reduced branching via constrained action subsets $\mathcal{A}_C(s)$. \\ \addlinespace

G5 & Update policies and causal estimates over time while preserving structural constraints on admissible action sequences. & Periodic or drift-triggered SCM re-estimation; continual reinforcement learning updates; recalibration of escalation thresholds under distributional shift. \\

\bottomrule
\end{tabularx}
\end{table}

\section{The C-MADF Framework}
\label{sec:cmadf_framework}

The Causal Multi-Agent Decision Framework (C-MADF) is a structured architecture for verifiable and explainable autonomous cyber defense. The framework integrates causal modeling, constrained decision processes, dual-objective multi-agent reinforcement learning with asymmetric reward shaping, and human-in-the-loop oversight into a unified pipeline. Its design objective is to ensure that every autonomous intervention is (i) causally grounded, (ii) restricted to formally admissible investigation trajectories, and (iii) accompanied by machine-generated explanations suitable for operational auditing.\\

Section 5 presents C-MADF in a coherent execution sequence rather than as isolated modules. It begins with Causal Discovery (Subsection~\ref{subsec:causal_discovery}), followed by the Causally-Constrained MDP-DAG Roadmap (Subsection~\ref{subsec:mdp_dag}), then the Council of Rivals deliberation (Subsection~\ref{subsec:council_of_rivals}), and finally the Explainable HITL interface (Subsection~\ref{subsec:hitl_interface}). Subsection~\ref{subsec:overall_algorithm} unifies these components into a complete end-to-end incident-response loop.


Table~\ref{tab:cmadf_components} presents the module inputs, outputs, and guarantees, while Figure~\ref{fig:cmadf_architecture} illustrates how these modules are connected in the pipeline. The surrounding text focuses on how outputs from each module become inputs to the next stage.

\begin{table}[t]
\centering
\caption{C-MADF Architectural Modules and Formal Guarantees}
\label{tab:cmadf_components}
\begin{tabularx}{\columnwidth}{@{} p{0.22\columnwidth} X X @{}}
\toprule
\textbf{Module} & \textbf{Inputs / Outputs} & \textbf{Formal Guarantee} \\ 
\midrule

Causal Discovery Engine 
& \textbf{In:} historical telemetry and event logs $\mathcal{D}$; 
\textbf{Out:} weighted causal DAG $\dagG=(V,E)$ with edge confidence scores. 
& \textit{Causal identifiability:} filters spurious statistical dependencies and enables intervention- and counterfactual-consistent reasoning under structural constraints. \\

Causally-Constrained MDP (MDP-DAG) 
& \textbf{In:} $\dagG$ and investigation ontology; 
\textbf{Out:} constrained MDP $\mathcal{M}=(S,A,P,R)$ over admissible roadmap states and actions. 
& \textit{Semantic soundness:} restricts policy exploration to transitions consistent with causal structure and operational semantics, ensuring verifiable action admissibility. \\

Council of Rivals Deliberation 
& \textbf{In:} current state $s_t$ and roadmap constraints; 
\textbf{Out:} adversarially validated action $\hat{a}_t$ and policy divergence score $\mathcal{D}(s_t)$. 
& \textit{Robust decision validation:} dual-policy asymmetric self-play with objective divergence reduces reasoning drift, mitigates overconfidence, and exposes epistemic uncertainty via calibrated disagreement. \\

Explainable Human-in-the-Loop (HITL) Interface 
& \textbf{In:} deliberation artifacts (evidence traces, counterarguments, $\mathcal{D}(s_t)$); 
\textbf{Out:} structured explanations, Explanation Transparency Score (ETS$(s_t)$), and escalation triggers. 
& \textit{Operational transparency:} produces operator-aligned rationales with quantifiable reliability calibration and principled escalation under high uncertainty. \\

\bottomrule
\end{tabularx}
\end{table}

\begin{figure*}[h!]
    \centering
    \includegraphics[width=0.9\textwidth]{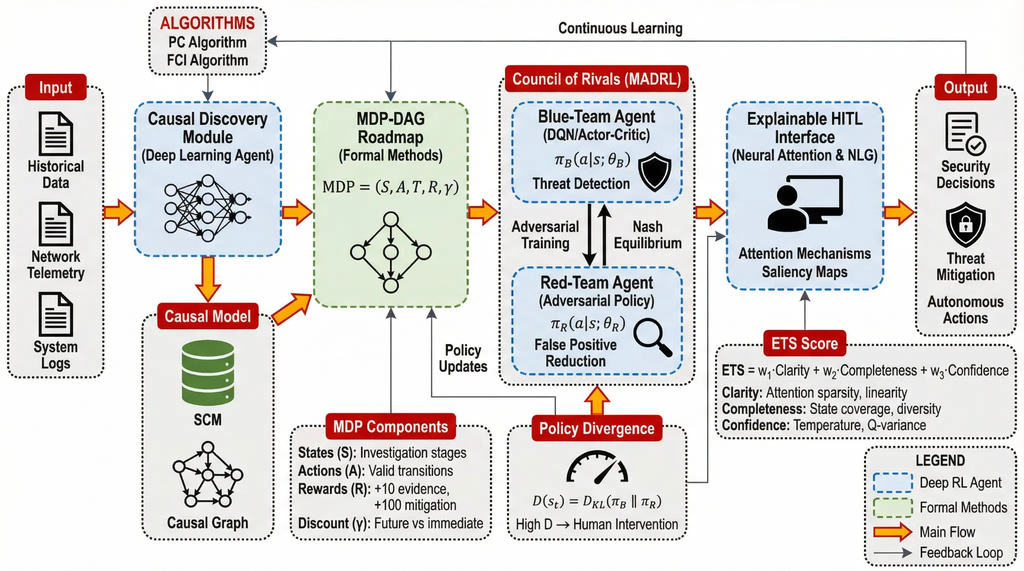}
    \caption{The architecture of the Causal Multi-Agent Decision Framework (C-MADF). The process begins with the Causal Discovery Module learning a causal model from data. This model informs the MDP-DAG Roadmap, which provides a verifiable structure for investigations. The Council of Rivals, consisting of a Blue-Team and a Red-Team agent, deliberates on the best course of action within this roadmap. Their debate and the resulting policy divergence are fed into the Explainable Human-in-the-Loop Interface, which computes the ETS and presents a clear recommendation to the supervisory adjudication.}
    \label{fig:cmadf_architecture}
\end{figure*}

\subsection{Causal Discovery Module}
\label{subsec:causal_discovery}

C-MADF is grounded in an explicit SCM of the monitored environment. Unlike correlation-based analytics pipelines, which may confound co-occurrence with causation, this module seeks to recover directed causal dependencies among security-relevant variables.

Let $\mathcal{V}$ denote the set of observed variables extracted from host- and network-level telemetry, including process metadata, system calls, authentication events, and flow-level features. The Causal Discovery Module applies constraint-based causal structure learning algorithms (e.g., PC or FCI) to estimate a DAG $\dagG = (V, E)$ over $\mathcal{V}$. Edge existence and orientation are inferred via conditional independence testing under standard causal sufficiency or partial observability assumptions.

The learned DAG is interpreted as the graphical representation of an SCM $M = \langle U, V, f \rangle$, where $U$ denotes exogenous variables and $f$ specifies structural equations governing endogenous variables $V$. The SCM enables intervention reasoning through $do(\cdot)$ operators and supports counterfactual queries over alternative action sequences.\\

To accommodate evolving infrastructure and adversarial tactics, the causal model is periodically re-estimated or incrementally updated under drift detection triggers. This ensures that structural dependencies reflect current operational conditions rather than stale historical patterns.

The SCM contributes three tightly coupled capabilities: causal filtering to suppress spurious correlations from coincidental co-occurrence or adversarial noise, intervention reasoning to evaluate the downstream effects of containment or blocking actions, and counterfactual analysis to compare alternative response strategies during post-incident auditing.

Algorithm~\ref{alg:causal_discovery} abstracts a family of constraint-based methods (such as PC/FCI). In practice, conditional independence tests are instantiated using suitable tests for mixed data types (e.g., partial correlation, kernel-based tests), and the algorithm is periodically re-run or incrementally updated to track concept drift in the security environment.

\begin{algorithm}[t!]
\caption{Constraint-Based Causal Discovery from Telemetry (PC Skeleton + Meek Orientation)}
\label{alg:causal_discovery}
\begin{algorithmic}[1]
\REQUIRE Historical telemetry dataset $\mathcal{D}$ over variables $V$; conditional independence test $\textsc{CI}(\cdot)$; significance level $\alpha$; maximum conditioning size $L_{\max}$; (optional) background knowledge $\mathcal{K}$.
\ENSURE Estimated causal DAG $\widehat{\mathcal{G}}=(V,\widehat{E})$.

\STATE \textbf{Preprocess} $\mathcal{D}$ (time alignment, missing-value handling, normalization/aggregation); construct variable set $V$.
\STATE Initialize an undirected complete graph $G \leftarrow (V,E)$ where $E=\{\{X,Y\}\mid X\neq Y,\ X,Y\in V\}$.
\STATE Initialize separating sets $\text{Sep}(X,Y)\leftarrow \emptyset$ for all unordered pairs $(X,Y)$.

\STATE $l \leftarrow 0$ \COMMENT{conditioning-set size}
\WHILE{$l \le L_{\max}$}
    \STATE $E_{\text{changed}} \leftarrow \textbf{false}$
    \FORALL{adjacent pairs $(X,Y)$ in $G$}
        \STATE $\text{Adj}(X)\leftarrow \{W \in V \mid \{X,W\}\in E\}$
        \IF{$|\text{Adj}(X)\setminus\{Y\}| < l$}
            \STATE \textbf{continue}
        \ENDIF
        \FORALL{$Z \subseteq \text{Adj}(X)\setminus\{Y\}$ with $|Z|=l$}
            \IF{$\textsc{CI}(X,Y \mid Z;\mathcal{D},\alpha)=\textbf{independent}$}
                \STATE Remove edge $\{X,Y\}$ from $G$.
                \STATE $\text{Sep}(X,Y)\leftarrow Z$; $\text{Sep}(Y,X)\leftarrow Z$.
                \STATE $E_{\text{changed}} \leftarrow \textbf{true}$
                \STATE \textbf{break} \COMMENT{no need to test larger $Z$ once separated}
            \ENDIF
        \ENDFOR
    \ENDFOR
    \IF{$E_{\text{changed}}=\textbf{false}$}
        \STATE \textbf{break}
    \ENDIF
    \STATE $l \leftarrow l + 1$
\ENDWHILE

\STATE \textbf{Initialize orientations:} replace each undirected edge $\{X,Y\}$ with $X{-}Y$.
\STATE \textbf{Collider identification (v-structures):}
\FORALL{triples $(X,Z,Y)$ such that $X{-}Z$ and $Y{-}Z$ and $X$ and $Y$ are non-adjacent in $G$}
    \IF{$Z \notin \text{Sep}(X,Y)$}
        \STATE Orient $X \rightarrow Z \leftarrow Y$.
    \ENDIF
\ENDFOR

\STATE \textbf{Apply orientation propagation rules} (Meek's rules) to orient remaining edges while enforcing acyclicity and consistency with $\mathcal{K}$.
\STATE \textbf{return} $\widehat{\mathcal{G}}$
\end{algorithmic}
\end{algorithm}

The specific implementation choices for this module are detailed in Table~\ref{tab:causal_discovery_config}. 
By specifying the discovery paradigm (PC/FCI family) and significance thresholds, 
the table provides the technical rationale for how the system balances structural 
accuracy with computational complexity. This configuration is vital for ensuring 
the learned DAG represents a stable approximation of the underlying security environment.

\begin{table}[t]
\centering
\caption{Causal Discovery Configuration and Design Choices}
\label{tab:causal_discovery_config}
\begin{tabularx}{\columnwidth}{@{} l X @{}}
\toprule
\textbf{Component} & \textbf{Specification and Rationale} \\ \midrule

Discovery paradigm &
Constraint-based causal structure learning (PC/FCI family), selected for its ability to produce an explicit, interpretable DAG consistent with observed conditional-independence relations in telemetry data. FCI is employed when latent confounding cannot be excluded. \\

Conditional independence testing &
Hybrid CI testing pipeline: partial correlation tests for approximately Gaussian continuous variables; kernel-based or mutual-information–based tests for mixed or non-Gaussian variables. Tests are selected based on feature-type metadata and validated via hold-out consistency checks. \\

Significance threshold &
Statistical significance level $\alpha \in [0.01, 0.05]$, selected through validation to balance false edge removal (Type I error) and spurious dependency retention (Type II error). Sensitivity analysis is performed across this range. \\

Maximum conditioning set size &
Conditioning set size bounded (e.g., $|Z| \le 2$--$4$) to control computational complexity and reduce instability in high-dimensional telemetry. This bounds worst-case complexity of independence testing and mitigates overfitting. \\

Edge orientation constraints &
Standard v-structure orientation and acyclicity enforcement, optionally augmented with domain priors (e.g., host process creation precedes outbound network flow). Prior constraints are encoded as forbidden or required edges. \\

Concept drift handling &
Periodic offline re-estimation or incremental updates triggered when distributional shift exceeds a predefined divergence threshold (e.g., population stability index or KL divergence). Drift detection prevents structural staleness. \\

Outputs &
Learned causal DAG $\dagG = (V,E)$ with associated edge confidence scores. Edge confidence values are propagated to (i) constrain admissible transitions in the MDP-DAG roadmap and (ii) modulate epistemic uncertainty in ETS computation. \\

\bottomrule
\end{tabularx}
\end{table}

\subsection{Causally-Constrained MDP-DAG Roadmap}
\label{subsec:mdp_dag}

The SCM learned in Section~\ref{subsec:causal_discovery} is compiled into an investigation-level roadmap that governs admissible response trajectories. The roadmap is distinct from the variable-level causal DAG: the SCM encodes dependencies among telemetry variables, whereas the roadmap DAG encodes permissible transitions among abstract investigation states. The graphical representation can be found in  Figure \ref{fig:mdp_dag_roadmap}.

The formal constrained MDP tuple and symbol definitions are introduced in Section 3 (Preliminaries) and are not repeated here. In this subsection, we focus on how the learned SCM instantiates that formulation through a roadmap-level action-admissibility structure.

Subsection~\ref{subsec:causal_discovery} provides the learned causal graph consumed by this stage. The output of this stage is a state-conditioned admissible action set used directly by policy optimization and online execution gating.

Let $\mathcal{G}_{\text{road}} = (\mathcal{S}, E_{\text{road}})$ denote the investigation roadmap DAG. 
Each edge $(s \rightarrow s') \in E_{\text{road}}$ is associated with a confidence score 
$c(s,s') \in [0,1]$ derived from the underlying causal discovery procedure.

The admissible action set at state $s$ is defined as:
\[
\mathcal{A}(s) = \{ a \in \mathcal{A} \mid (s \rightarrow s') \in E_{\text{road}} \}.
\]

In practice, we distinguish between high-confidence and low-confidence transitions. 
Transitions whose confidence exceeds a threshold $\tau_c$ are treated as \emph{hard structural constraints}, whereas lower-confidence transitions are treated as \emph{soft constraints}. 
Soft-constrained transitions are not permanently eliminated; instead, they trigger elevated epistemic monitoring and potential human escalation rather than structural prohibition.

Edges in $\mathcal{G}_{\text{road}}$ are derived by projecting intervention-consistent transitions from the learned SCM. 
This construction prioritizes transitions that respect prerequisite evidence-gathering stages implied by the causal structure, while allowing calibrated fallback escalation when structural uncertainty is high.

The reward function $\mathcal{R}(s,a)$ encodes operational objectives through a weighted aggregation of mitigation success, disruption cost, resource expenditure, and evidential utility. Weight calibration is deployment-specific and performed via validation on historical incident traces.

\begin{figure*}[h!]
    \centering
    \includegraphics[width=\textwidth]{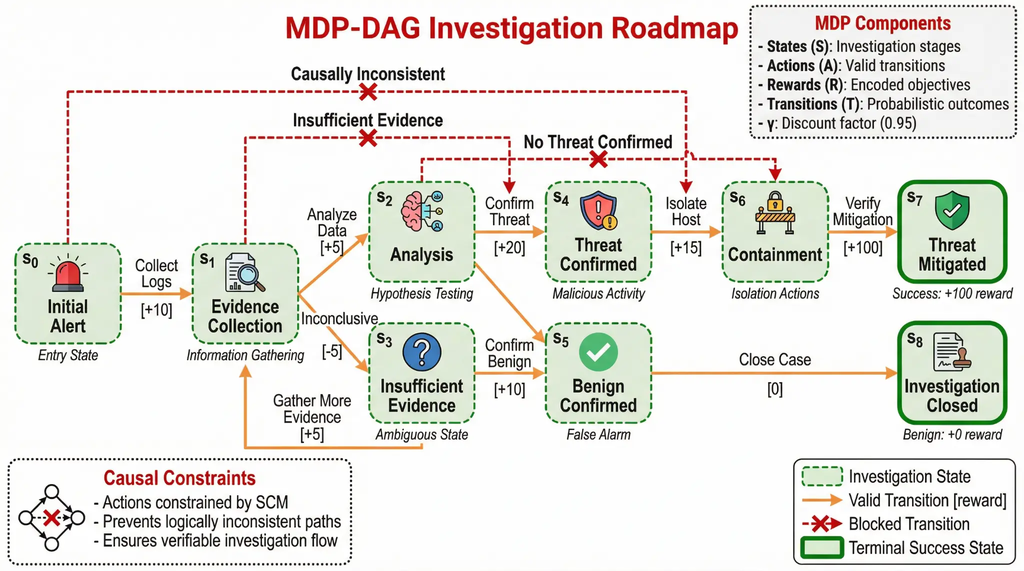}
    \caption{The MDP-DAG Investigation Roadmap showing valid state transitions (orange arrows with reward values) and causally inconsistent blocked transitions (red X). The roadmap constrains the action space to ensure verifiable investigation flows, preventing logically inconsistent paths such as jumping directly from Initial Alert to Threat Mitigated without sufficient evidence collection.}
    \label{fig:mdp_dag_roadmap}
\end{figure*}

\begin{algorithm}[t]
\caption{MDP-DAG Construction from a Learned SCM}
\label{alg:mdp_dag_construction}
\begin{algorithmic}[1]
\STATE \textbf{Input:} Learned SCM $M = \langle U, V, f \rangle$ with causal DAG $\dagG = (V,E)$; investigation ontology $\mathcal{O}$
\STATE \textbf{Output:} Constrained MDP $\mathcal{M} = (\mathcal{S}, \mathcal{A}, \mathcal{P}, \mathcal{R}, \gamma)$

\vspace{0.3em}
\STATE \textbf{Step 1: State Space Construction}
\STATE Define investigation states $\mathcal{S}$ as abstract stages derived from ontology $\mathcal{O}$ (e.g., \emph{Alert}, \emph{Evidence Gathering}, \emph{Containment}, \emph{Resolution}).
\STATE Each state $s \in \mathcal{S}$ is associated with a subset of observed variables $V_s \subseteq V$ required to justify that stage.

\vspace{0.3em}
\STATE \textbf{Step 2: Admissible Transition Identification}
\STATE Initialize empty roadmap DAG $\mathcal{G}_{\text{road}} = (\mathcal{S}, E_{\text{road}})$.
\FOR{each ordered pair $(s_i, s_j) \in \mathcal{S} \times \mathcal{S}$}
    \IF{$s_j$ represents a valid semantic progression from $s_i$ under ontology $\mathcal{O}$}
        \IF{all required causal parents of $V_{s_j}$ are satisfied along $\dagG$}
            \STATE Add edge $(s_i, s_j)$ to $E_{\text{road}}$.
            \STATE Define action $a_{ij}$ corresponding to transition $(s_i \rightarrow s_j)$.
        \ENDIF
    \ENDIF
\ENDFOR

\vspace{0.3em}
\STATE \textbf{Step 3: Transition and Reward Specification}
\STATE Define $\mathcal{A}$ as the set of all admissible transition actions $a_{ij}$.
\STATE Estimate transition probabilities $\mathcal{P}(s' \mid s, a)$ from historical incident traces or benchmark telemetry.
\STATE Define reward function $\mathcal{R}(s,a)$ as weighted aggregation of mitigation success, disruption cost, resource usage, and evidential gain.

\vspace{0.3em}
\STATE Select discount factor $\gamma \in (0,1)$ to balance mitigation urgency and investigation thoroughness.

\STATE \textbf{return} $\mathcal{M} = (\mathcal{S}, \mathcal{A}, \mathcal{P}, \mathcal{R}, \gamma)$
\end{algorithmic}
\end{algorithm}

\subsection{Adversarial Multi-Agent Deliberation}
\label{subsec:council_of_rivals}

Given the admissible action set produced by the roadmap stage, decision-making is performed by a two-agent multi-agent reinforcement learning (MARL) architecture that introduces structured counterfactual scrutiny into autonomous response selection\footnote{We note that the Blue-Team and Red-Team policies do not constitute a classical zero-sum adversarial game. Rather, they represent asymmetric objective functions defined over a shared constrained MDP. The Red-Team does not model the external attacker; instead, it serves as a conservative counterfactual policy regularizer that penalizes premature mitigation under evidentiary ambiguity.}.

Let $\pi_B(a \mid s; \theta_B)$ denote the Blue-Team policy and $\pi_R(a \mid s; \theta_R)$ the Red-Team policy. Both policies operate over the admissible action set $\mathcal{A}(s)$ defined by the roadmap DAG (Section~\ref{subsec:mdp_dag}). The interaction between the two policies is modeled as a stochastic game defined over the same state space $\mathcal{S}$ and constrained transition structure.\\

\noindent \textbf{Blue-Team Policy.}
The Blue-Team policy is optimized to maximize cumulative mitigation utility under the reward function $\mathcal{R}_B$, which prioritizes successful threat neutralization, evidentiary progression, and investigation efficiency subject to the admissibility constraints. The policy may be instantiated using deep Q-learning, actor–critic methods, or other function-approximation schemes compatible with constrained action spaces.

\noindent \textbf{Red-Team Policy.}
The Red-Team policy is optimized under a conservatively shaped reward $\mathcal{R}_R$ that penalizes unjustified or high-impact interventions lacking sufficient evidentiary support. Its objective is to minimize false positives and operational disruption. Unlike classical adversarial learning that perturbs inputs, the Red-Team operates at the decision-process level by proposing alternative admissible actions within $\mathcal{A}(s)$.\\

\noindent \textbf{Game-Theoretic Interaction.}
At each state $s_t$, both agents evaluate the admissible action set:
\[
a_B \sim \pi_B(\cdot \mid s_t), \quad
a_R \sim \pi_R(\cdot \mid s_t).
\]
The resulting joint evaluation constitutes a stage of a two-player stochastic game with aligned state dynamics but partially opposing objectives. Training may be conducted via self-play or alternating optimization until convergence to a stable policy pair.\\

The interaction between these two agents is visualized in Figure~\ref{fig:council_of_rivals}. 
The figure shows the iterative process of hypothesis evaluation and action arbitration. 
By forcing the Blue-Team to justify its proposals against a skeptical Red-Team, 
the framework creates a robust decision-validation loop that exposes epistemic 
uncertainty as a measurable divergence score.\\

\noindent \textbf{Policy Divergence as Epistemic Signal.}
To quantify disagreement between agents, we define a Policy Divergence Score:
\begin{equation}
\mathcal{D}(s_t) = D\bigl(\pi_B(\cdot \mid s_t), \pi_R(\cdot \mid s_t)\bigr),
\end{equation}
where $D(\cdot,\cdot)$ denotes a distributional divergence metric (e.g., KL divergence or total variation distance). For value-based implementations, divergence may alternatively be defined as
\begin{equation}
\mathcal{D}(s_t) = \left\| Q_B(s_t, \cdot) - Q_R(s_t, \cdot) \right\|_2.
\end{equation}

Large divergence indicates disagreement in expected utility estimates under competing objectives and is interpreted as elevated epistemic uncertainty. When $\mathcal{D}(s_t)$ exceeds a predefined threshold, autonomous execution is suspended and the decision is escalated to the human-in-the-loop interface.\\

\noindent \textbf{Structural Role.}
Because both agents operate within the same causally constrained action space, structured counterfactual scrutiny cannot introduce logically inconsistent transitions. Instead, disagreement is confined to admissible investigation paths, ensuring that robustness is achieved at the policy-selection level without violating structural constraints.

This adversarial deliberation mechanism mitigates unilateral commitment to high-impact actions and reduces susceptibility to telemetry manipulation by requiring cross-policy consistency prior to execution.\\

\begin{table}[t]
\centering
\caption{Reward Shaping Design for the Council of Rivals}
\label{tab:reward_design}
\begin{tabularx}{\columnwidth}{@{} >{\hsize=0.7\hsize}X >{\hsize=1.15\hsize}X >{\hsize=1.15\hsize}X @{}}
\toprule
\textbf{Reward Component} & \textbf{Blue-Team Policy ($\pi_B$)} & \textbf{Red-Team Policy ($\pi_R$)} \\ \midrule

Verified mitigation 
& Large positive reward for reaching causally justified terminal states (e.g., ``Threat Mitigated'') after traversing required evidence states. 
& Small positive reward only if mitigation is supported by sufficient evidence; no reward for premature containment. \\

False-positive intervention 
& Penalty proportional to operational disruption cost (e.g., unnecessary isolation, service interruption). 
& Large positive reward for preventing or opposing unjustified disruptive actions; penalty if false positives are not challenged. \\

Evidence acquisition 
& Positive reward for actions that increase information gain or reduce causal uncertainty along the learned DAG. 
& Positive reward for demanding additional evidence under high epistemic uncertainty; penalty for accepting unsupported mitigation claims. \\

Investigation efficiency 
& Small step penalty to discourage unnecessary actions and encourage timely resolution. 
& Smaller penalty weight relative to Blue-Team, prioritizing verification over speed in ambiguous states. \\

Roadmap compliance 
& Penalty for proposing actions inconsistent with MDP-DAG constraints. 
& Equivalent penalty for violating causal or roadmap admissibility constraints. \\

Human escalation 
& Neutral or small penalty when escalation is triggered under high divergence $\mathcal{D}(s_t)$. 
& Positive reward when escalation is correctly triggered under high epistemic uncertainty. \\

\bottomrule
\end{tabularx}
\end{table}

\begin{figure*}[h!]
    \centering
    \includegraphics[width=0.9\textwidth]{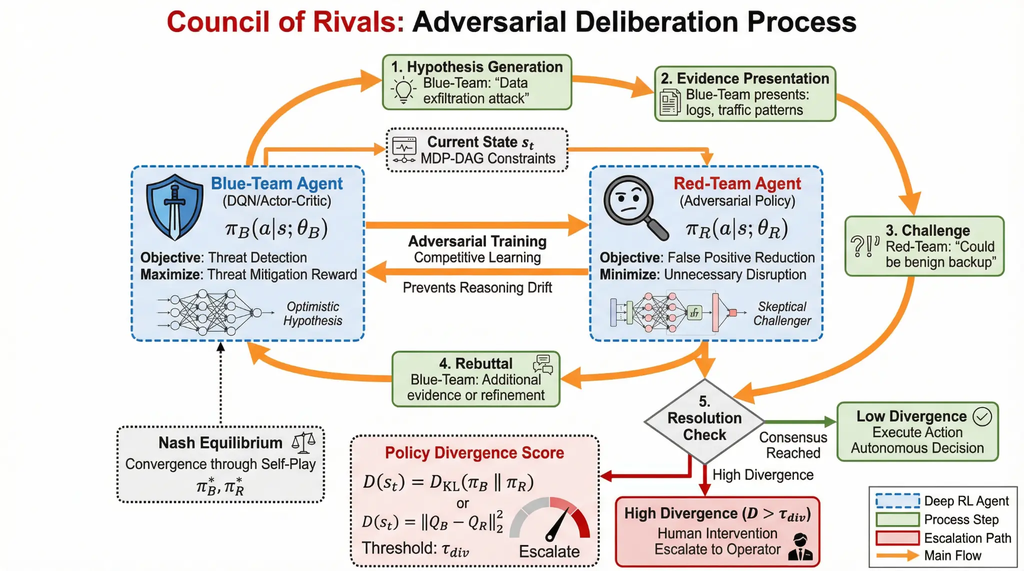}
    \caption{Council of Rivals adversarial deliberation architecture.
    At each investigation state $s_t$ in the causally constrained MDP-DAG, the Blue-Team policy proposes a mitigation action based on its learned value estimates, while the Red-Team policy evaluates and challenges the proposal under a conservatively shaped objective. The interaction proceeds through iterative hypothesis evaluation, evidential justification, and action arbitration within the constrained action space. Inter-policy disagreement is quantified via the Policy Divergence Score $\mathcal{D}(s_t)$, which serves as an epistemic uncertainty signal and triggers human escalation when $\mathcal{D}(s_t) > \tau_{\text{div}}$. The policies are trained via self-play in a two-player stochastic game defined over the MDP-DAG, yielding stabilized adversarial decision policies.}
    \label{fig:council_of_rivals}
\end{figure*}

\begin{algorithm}[t]
\caption{Council of Rivals Self-Play Training}
\label{alg:council_training}
\begin{algorithmic}[1]
\STATE \textbf{Input:} Constrained MDP-DAG $\mathcal{M} = (\mathcal{S},\mathcal{A},\mathcal{P},\mathcal{R},\gamma)$
\STATE \textbf{Initialize:} Blue policy $\pi_B(\cdot;\theta_B)$, Red policy $\pi_R(\cdot;\theta_R)$, target networks $\hat{\theta}_B, \hat{\theta}_R$
\STATE Initialize replay buffers $\mathcal{B}_B, \mathcal{B}_R$

\FOR{each training episode}
    \STATE Sample initial investigation state $s_0 \sim \rho_0$
    \FOR{$t = 0$ to $T-1$}
        \STATE Apply action mask $\mathcal{A}_{\text{valid}}(s_t)$ from MDP-DAG constraints
        \STATE Blue proposes $a_B \sim \pi_B(\cdot \mid s_t; \theta_B)$
        \STATE Red proposes critique/alternative $a_R \sim \pi_R(\cdot \mid s_t; \theta_R)$
        \STATE Resolve executed action:
        \[
            a_t \leftarrow \mathrm{Arbitrate}(a_B, a_R, s_t)
        \]
        \STATE Execute $a_t$, observe $s_{t+1}$ and rewards $(r_t^B, r_t^R)$
        \STATE Store $(s_t, a_B, a_R, a_t, r_t^B, r_t^R, s_{t+1})$ in buffers
        \STATE Sample mini-batch from $\mathcal{B}_B$ and update $\theta_B$ via Double DQN:
        \[
        y = r_t^B + \gamma Q_{\hat{\theta}_B}(s_{t+1}, 
        \arg\max_a Q_{\theta_B}(s_{t+1},a))
        \]
        \STATE Update $\theta_R$ analogously using $r_t^R$
        \STATE Update target networks via Polyak averaging:
        \[
        \hat{\theta} \leftarrow \tau \theta + (1-\tau)\hat{\theta}
        \]
    \ENDFOR
\ENDFOR

\STATE \textbf{return} $\pi_B^\star, \pi_R^\star$
\end{algorithmic}
\end{algorithm}

Algorithm~\ref{alg:council_training} summarizes the self-play training procedure. The reward structure may be zero-sum, general-sum, or shaped to encourage calibrated skepticism and robust consensus, depending on deployment requirements.\\

\subsubsection{Policy Architecture and Training Configuration}
\label{subsec:policy_architecture}

The Blue-Team and Red-Team policies share a common function-approximation backbone and optimization framework. Their behavioral divergence arises exclusively from reward shaping and objective weighting (Table~\ref{tab:reward_design}), ensuring that observed differences in action selection are attributable to decision criteria rather than architectural asymmetry.\\

\noindent \textbf{State Representation.} Each investigation state $s_t \in \mathcal{S}$ is mapped to a fixed-dimensional embedding $x_t \in \mathbb{R}^d$. The feature vector aggregates three components:

\begin{itemize}
    \item \textbf{Host-level statistics} aggregated over a sliding window of length $K_h$, including process-level activity counts, resource utilization metrics, and security event frequencies;
    \item \textbf{Network-level statistics} aggregated over window $K_n$, including flow volumes, endpoint diversity, protocol distribution, and connection dynamics;
    \item \textbf{Investigation context features}, including current roadmap node, evidence flags, mitigation history, and escalation indicators encoded as binary or multi-hot vectors.
\end{itemize}

The resulting representation is
\[
x_t = [h_t \,\|\, n_t \,\|\, c_t],
\]
with per-feature standardization performed using parameters estimated from a held-out calibration dataset. This abstraction assumes that the investigation state is fully observable at the roadmap level, with partial observability in raw telemetry absorbed into feature aggregation.\\

\noindent \textbf{Action Space and Masking.} The discrete action space $\mathcal{A}$ comprises 18 high-level investigation and mitigation actions aligned with the roadmap DAG. At state $s_t$, admissible actions are restricted to $\mathcal{A}(s_t)$ via action masking consistent with the structural constraints defined in Section~\ref{subsec:mdp_dag}. Invalid transitions are assigned zero probability prior to policy evaluation, ensuring that learning occurs exclusively within the causally admissible state–action subspace.\\

\noindent \textbf{Function Approximation.} Both agents employ a three-layer multilayer perceptron (MLP) with ReLU activations and layer normalization:

\[
\phi(s_t) = \mathrm{LayerNorm}\left(
\sigma\left(W_3 \sigma\left(W_2 \sigma\left(W_1 x_t + b_1\right) + b_2\right) + b_3\right)
\right),
\]
where $\sigma(\cdot)$ denotes the ReLU activation function.

For value-based training, the action-value function is defined as
\[
Q(s_t, a; \theta) = f_Q(\phi(s_t), a),
\]
where $f_Q$ denotes a linear or bilinear projection from the shared representation to action scores. In actor–critic configurations, the shared backbone $\phi(\cdot)$ feeds separate heads for policy logits and state-value estimation.

Parameter sharing across agents is limited to architectural structure; weights $\theta_B$ and $\theta_R$ are trained independently.\\

\noindent \textbf{Training Procedure and Stability Mechanisms.} Training is conducted using off-policy deep Q-learning with target networks and prioritized experience replay. Each agent maintains its own replay buffer to preserve objective asymmetry. To enhance stability, we employ:

\begin{itemize}
    \item target-network updates via Polyak averaging with coefficient $\tau = 0.005$;
    \item Double DQN updates to mitigate overestimation bias;
    \item $\ell_2$ gradient clipping with threshold 5.0;
    \item $\varepsilon$-greedy exploration annealed from 1.0 to 0.05 over $3 \times 10^5$ environment steps.
\end{itemize}

Training convergence is monitored using (i) exponentially smoothed episodic return and (ii) validation-set false-positive rate under adversarial perturbation. Early stopping is triggered when both metrics stabilize over 20 consecutive evaluation intervals.

All hyperparameters are fixed across experimental conditions unless otherwise specified in Section~\ref{sec:evaluation}.














\begin{table}[t]
\centering
\caption{Council of Rivals Training Configuration}
\label{tab:training_hparams}
\begin{tabularx}{\columnwidth}{@{} >{\hsize=0.9\hsize}X >{\hsize=0.7\hsize}X >{\hsize=1.4\hsize}X @{}}
\toprule
\textbf{Parameter} & \textbf{Value} & \textbf{Specification} \\ \midrule

State dimension $d$ 
& 256 
& Resulting dimensionality after concatenation and linear projection of host, network, and investigation-context features. \\

Hidden layer sizes 
& [256, 256, 128] 
& Shared multilayer perceptron backbone with ReLU activations and Layer Normalization. \\

Optimizer 
& Adam 
& $\beta_1=0.9$, $\beta_2=0.999$, default $\epsilon=10^{-8}$. \\

Learning rate 
& $3\times10^{-4}$ 
& Selected via validation-based hyperparameter search. \\

Batch size 
& 128 
& Mini-batches sampled from prioritized replay buffer. \\

Replay buffer capacity 
& $5\times10^5$ transitions 
& Per-agent buffer with proportional prioritized sampling. \\

Discount factor $\gamma$ 
& 0.99 
& Supports long-horizon investigation planning. \\

Target update coefficient $\tau$ 
& 0.005 
& Polyak averaging for stabilizing target networks. \\

Exploration schedule 
& $\varepsilon$-greedy: $1.0 \rightarrow 0.05$ 
& Linear decay over $3\times10^5$ environment steps. \\

Total training steps 
& $1.5\times10^6$ 
& Per agent across all training scenarios. \\

Random seeds 
& \{1,2,3\} 
& Results reported as mean across three independent runs. \\

\bottomrule
\end{tabularx}
\end{table}

\begin{table}[t]
\centering
\caption{Decomposition of the ETS}
\label{tab:ets_components}
\begin{tabularx}{\columnwidth}{@{} l X @{}}
\toprule
\textbf{Component} & \textbf{Operational Definition} \\ \midrule
Clarity 
& Concentration and stability of explanatory signals, quantified via (i) entropy of feature-level attention distributions, (ii) local linear fidelity of the Blue-Team value function $Q_B$, and (iii) stability of action rankings under bounded state perturbations. \\

Completeness 
& Evidential sufficiency and coverage, measured by (i) replay-buffer state density in a neighborhood of $s_t$, (ii) entropy of empirically observed action distributions in similar historical states, and (iii) diversity of feature-attribution contributions. \\

Confidence 
& Epistemic certainty of the recommended action $a^\star$, estimated through (i) inverse effective policy temperature $\tau_{\pi_B}$, (ii) bootstrap variance of $Q_B(s_t, a^\star)$, and (iii) a monotonic transformation of the inter-agent Policy Divergence Score $\mathcal{D}(s_t)$. \\

Aggregation 
& Weighted combination 
$\text{ETS}(s_t) = w_1 \cdot \text{Clarity} + w_2 \cdot \text{Completeness} + w_3 \cdot \text{Confidence}$, followed by bounded normalization to $[0,1]$ for use in escalation gating. \\

\bottomrule
\end{tabularx}
\end{table}

\begin{figure*}[h!]
    \centering
    \includegraphics[width=0.9\textwidth]{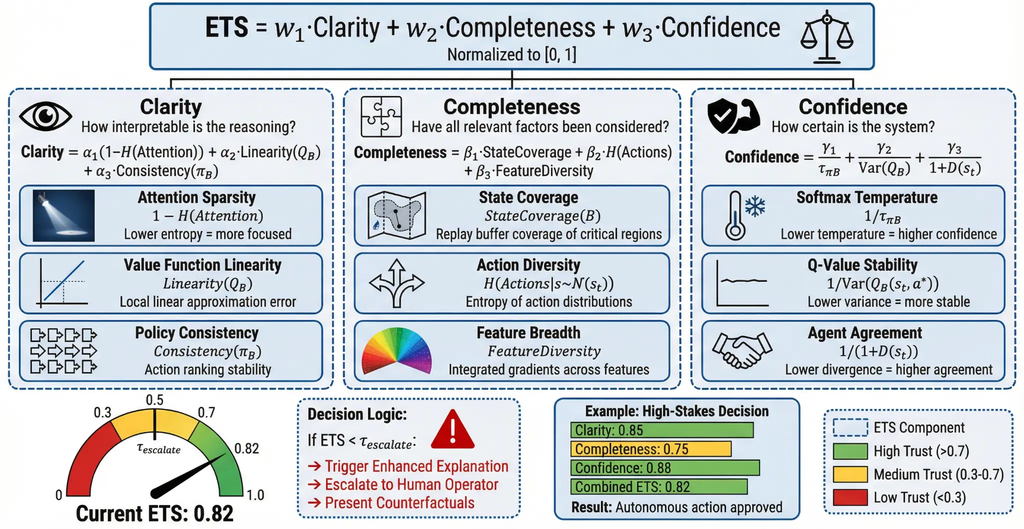}
\caption{Illustrative decomposition of the ETS into its three primary components, Clarity, Completeness, and Confidence, and their respective sub-metrics. The figure visualizes the aggregation structure used to compute ETS$(s_t)$ and demonstrates how the resulting score is compared against the escalation threshold $\tau_{\text{escalate}}$. In this example, ETS$(s_t)=0.82$ exceeds the calibrated threshold, and the system proceeds without mandatory human escalation. When ETS$(s_t)$ falls below the threshold, structured explanation artifacts and human oversight are triggered.}
    \label{fig:ets_breakdown}
\end{figure*}

\subsection{Explainable Human-in-the-Loop Interface}
\label{subsec:hitl_interface}

C-MADF exposes its internal deliberation process through a structured human-in-the-loop (HITL) interface designed for supervisory control and post hoc auditability. Rather than presenting raw network outputs, this interface consolidates four decision artifacts in one view: competing admissible actions from $\pi_B$ and $\pi_R$, supporting/contradicting evidential features, inter-policy divergence $\mathcal{D}(s_t)$, and projected operational cost and mitigation impact for candidate actions.

The hierarchical structure of the ETS is decomposed into its constituent metrics in 
Table~\ref{tab:ets_components}, while Figure~\ref{fig:ets_breakdown} illustrates the aggregation logic. These visualizations explain how raw feature attributions and 
policy variances are transformed into a single, trust-calibrated indicator. This 
aggregation is essential for providing operators with a clear "go/no-go" signal 
during high-pressure incident response.

These artifacts are computed from policy distributions, value-function attributions, and structural constraints imposed by the MDP-DAG. The HITL layer does not modify policy learning; instead, it governs execution and escalation under uncertainty.\\

\noindent \textit{Explainability-Transparency Score (ETS).} To provide a bounded escalation signal, we define the ETS, a scalar meta-indicator aggregating explanation structure and epistemic certainty at state $s_t$.

\begin{equation}
\begin{aligned}
\text{ETS}(s_t) = &\, w_1 \,\text{Clarity}(s_t) \\
                  &+ w_2 \,\text{Completeness}(s_t) \\
                  &+ w_3 \,\text{Confidence}(s_t),
\end{aligned}
\end{equation}
where $w_i \ge 0$ and $\sum_i w_i = 1$. All sub-scores are normalized to $[0,1]$, ensuring $\text{ETS}(s_t) \in [0,1]$.

ETS functions as a calibrated decision-gating signal reflecting system-internal decision reliability rather than subjective human trust. Its parameters are selected to reflect operational interpretability and are empirically tuned using validation scenarios (Section~\ref{sec:evaluation}).\\

\noindent \textit{Clarity.} Clarity quantifies concentration and stability of explanatory signals associated with the recommended action $a^\star$.

Let $\alpha_t$ denote the normalized feature-attribution distribution for state $s_t$. Clarity incorporates:

\begin{itemize}
    \item entropy of feature attributions $H(\alpha_t)$,
    \item local linear fidelity of $Q_B$ under first-order approximation,
    \item ranking stability of top actions under small perturbations of $x_t$.
\end{itemize}

A representative formulation is:

\begin{equation}
\begin{aligned}
\text{Clarity}(s_t)
&= \alpha_1 \bigl(1 - \tfrac{H(\alpha_t)}{\log d}\bigr)
+ \alpha_2 \,\text{LocalFidelity}(s_t) \\
&\quad + \alpha_3 \,\text{RankingStability}(s_t),
\end{aligned}
\end{equation}
with $\alpha_i \ge 0$ and $\sum_i \alpha_i = 1$.\\

\noindent \textit{Completeness.} Completeness measures evidential coverage and behavioral diversity in the vicinity of $s_t$.

Let $\mathcal{B}$ denote the replay buffer and $\mathcal{N}_\epsilon(s_t)$ a local perturbation neighborhood defined via bounded feature noise $\|x - x_t\|_2 \le \epsilon$.

\begin{equation}
\begin{aligned}
\text{Completeness}(s_t)
&= \beta_1 \,\text{StateDensity}_{\mathcal{B}}(s_t)
+ \beta_2 \,\tfrac{H\bigl(\pi_B(\cdot \mid \mathcal{N}_\epsilon(s_t))\bigr)}{\log |\mathcal{A}|} \\
&\quad + \beta_3 \,\text{AttributionDiversity}(s_t),
\end{aligned}
\end{equation}
with $\beta_i \ge 0$ and $\sum_i \beta_i = 1$.\\

\noindent \textit{Confidence.} Confidence captures epistemic certainty in the recommended action $a^\star$.

\begin{equation}
\begin{aligned}
\text{Confidence}(s_t)
&= \gamma_1 \,\text{SoftmaxMargin}(s_t) \\
+ \gamma_2 \,\exp\!\bigl(-\text{Var}(Q_B(s_t,a^\star))\bigr) 
&\quad + \gamma_3 \,\exp\!\bigl(-\mathcal{D}(s_t)\bigr),
\end{aligned}
\end{equation}
with $\gamma_i \ge 0$ and $\sum_i \gamma_i = 1$.\\

\noindent \textit{Calibration.}
ETS parameters are calibrated using scenario-level correctness labels and evidentiary sufficiency heuristics derived from benchmark ground-truth annotations.

Each calibration scenario $s_j$ is associated with an oracle reliability target
$R_j \in [0,1]$, constructed from objective criteria including
(i) correctness of the recommended action with respect to ground-truth attack state,
(ii) causal trace consistency within the SCM,
and (iii) stability of policy outputs under bounded perturbations.

Calibration solves the constrained optimization problem:
\[
\min_{w} \sum_j \bigl(\text{ETS}(s_j; w) - R_j\bigr)^2
\quad \text{subject to} \quad w_i \ge 0, \ \sum_i w_i = 1.
\]

Calibration quality is evaluated using prediction error and monotonicity with respect to evidentiary sufficiency and policy agreement on held-out benchmark splits (Section~\ref{sec:evaluation}). This procedure treats ETS as a proxy for system-internal decision reliability rather than subjective human confidence.\\

\begin{algorithm*}[ht!]
\caption{ETS Computation}
\label{alg:ets_computation}
\begin{algorithmic}[1]

\STATE \textbf{Input:} 
State $s_t$, Blue/Red policies $(\pi_B, \pi_R)$, 
value estimator $Q_B$, replay buffer $\mathcal{B}$, 
calibrated weights $(w_1,w_2,w_3)$.

\STATE \textbf{Output:} $\mathrm{ETS}(s_t) \in [0,1]$

\vspace{0.3em}
\STATE \textbf{// Step 1: Compute Clarity component}
\STATE Extract feature-attribution or attention distribution $\alpha_t$.
\STATE Compute attention entropy $H(\alpha_t)$.
\STATE Estimate local linear approximation fidelity of $Q_B$ around $s_t$.
\STATE Evaluate action-ranking stability under small perturbations of $s_t$.
\STATE Aggregate into

\begin{align*}
\mathrm{Clarity}(s_t) = \ \alpha_1 (1 - H(\alpha_t)) 
                        + \alpha_2 \mathrm{Linearity}(Q_B) 
                        + \alpha_3 \mathrm{Consistency}(\pi_B).
\end{align*}

\vspace{0.3em}
\STATE \textbf{// Step 2: Compute Completeness component}
\STATE Estimate state-density in replay buffer neighborhood:
\[
\mathrm{StateCoverage}(s_t) = 
\frac{|\{s \in \mathcal{B} : \|s - s_t\| \le \epsilon\}|}{|\mathcal{B}|}.
\]
\STATE Compute empirical action entropy in local neighborhood.
\STATE Measure feature-attribution diversity across top-$k$ salient features.
\STATE Aggregate into
\[
\mathrm{Completeness}(s_t) =
\beta_1 \mathrm{StateCoverage} +
\beta_2 H(\text{Actions} \mid s \approx s_t) +
\beta_3 \mathrm{FeatureDiversity}.
\]

\vspace{0.3em}
\STATE \textbf{// Step 3: Compute Confidence component}
\STATE Let $a^\star = \arg\max_a \pi_B(a \mid s_t)$.
\STATE Estimate epistemic uncertainty via:
\begin{itemize}
\item inverse policy temperature $1/\tau_{\pi_B}$,
\item inverse bootstrap variance $\bigl(1+\mathrm{Var}(Q_B(s_t,a^\star))\bigr)^{-1}$,
\item inverse policy divergence $(1+\mathcal{D}(s_t))^{-1}$.
\end{itemize}
\STATE Aggregate into
\[
\mathrm{Confidence}(s_t) =
\gamma_1 \frac{1}{\tau_{\pi_B}} +
\gamma_2 \frac{1}{1+\mathrm{Var}(Q_B)} +
\gamma_3 \frac{1}{1+\mathcal{D}(s_t)}.
\]

\vspace{0.3em}
\STATE \textbf{// Step 4: Weighted aggregation and normalization}
\STATE Compute raw score:
\[
\mathrm{ETS}_{raw}(s_t) =
w_1 \mathrm{Clarity}(s_t) +
w_2 \mathrm{Completeness}(s_t) +
w_3 \mathrm{Confidence}(s_t).
\]
\STATE Apply bounded normalization (e.g., min–max or logistic scaling):
\[
\mathrm{ETS}(s_t) =
\sigma\!\left( \frac{\mathrm{ETS}_{raw}(s_t)-\mu}{\sigma_{cal}} \right).
\]

\STATE \textbf{return} $\mathrm{ETS}(s_t)$

\end{algorithmic}
\end{algorithm*}

\subsection{Overall C-MADF Decision Procedure}
\label{subsec:overall_algorithm}

We now summarize the complete online decision loop of C-MADF, integrating causal discovery, MDP-DAG planning, adversarial deliberation, and human oversight. Figure~\ref{fig:decision_loop} provides a high-level overview of the end-to-end decision process as formalized in Algorithm~\ref{alg:cmadf_overall}. It serves to synthesize the interaction between telemetry ingestion, causal roadmapping, and human-in-the-loop escalation, demonstrating how the framework maintains "Clarity of Command" across different operational phases.

\begin{algorithm}[ht!]
\caption{C-MADF Online Incident Response}
\label{alg:cmadf_overall}
\begin{algorithmic}[1]
\STATE \textbf{Offline / Periodic Phase:}
\STATE Collect and update historical telemetry $\mathcal{D}$ from the security environment.
\STATE Run Causal Discovery (Algorithm~\ref{alg:causal_discovery}) to obtain updated causal DAG $\dagG$.
\STATE Construct or refine MDP-DAG roadmap $\mathcal{M}$ using Algorithm~\ref{alg:mdp_dag_construction}.
\STATE Train or fine-tune Council of Rivals policies $\pi_B$, $\pi_R$ via self-play (Algorithm~\ref{alg:council_training}).

\vspace{0.5em}
\STATE \textbf{Online Incident Handling Phase:}
\FOR{each incoming alert or detected anomaly}
    \STATE Encode current context and telemetry into MDP state $s_0$.
    \FOR{$t = 0, 1, \dots, T_{\max}$}
        \STATE Blue-Team proposes action $a_B \sim \pi_B(\cdot \mid s_t)$.
        \STATE Red-Team critiques and proposes $a_R \sim \pi_R(\cdot \mid s_t)$.
        \STATE Compute policy divergence $\mathcal{D}(s_t)$.
        \STATE Resolve executed action $a_t$ based on $(a_B, a_R, \mathcal{D}(s_t))$ and safety rules.
        \STATE Simulate or execute $a_t$ on the environment (subject to safeguards), observe $s_{t+1}$ and reward $r_t$.
        \STATE Compute ETS$(s_t)$ using Algorithm~\ref{alg:ets_computation}.
        \IF{$\text{ETS}(s_t) < \tau_{\text{escalate}}$ \OR $\mathcal{D}(s_t) > \tau_{\text{div}}$}
            \STATE Generate detailed explanation: hypotheses, evidence, counterfactuals, and risk assessment.
            \STATE Present recommendation and rationale to supervisory adjudication for final decision.
            \STATE Log execution outcomes and system-internal reliability signals for offline analysis and threshold recalibration.

            \STATE \textbf{break} (end automated loop for this incident).
        \ELSIF{$s_{t+1}$ is a terminal state (e.g., ``Threat Mitigated'' or ``Investigation Aborted'')}
            \STATE Log trajectory, decisions, and outcomes for auditing and future training.
            \STATE \textbf{break}
        \ELSE
            \STATE Set $s_t \gets s_{t+1}$ and continue deliberation.
        \ENDIF
    \ENDFOR
\ENDFOR
\end{algorithmic}
\end{algorithm}

Algorithm~\ref{alg:cmadf_overall} provides a general, end-to-end specification of the C-MADF framework suitable for rigorous analysis and implementation. It ensures that all autonomous actions are grounded in causal structure, optimized via reinforcement learning, stress-tested by adversarial deliberation, and ultimately governed by human-understandable explanations and trust-calibrated intervention thresholds.

\begin{figure}[pht!]
    \centering
    \includegraphics[width=\columnwidth]{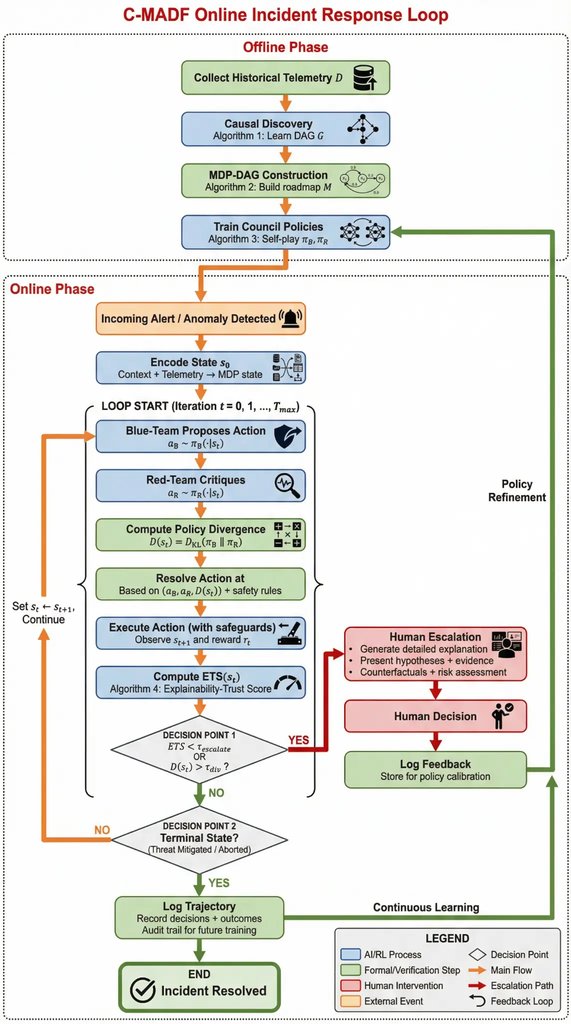}
\caption{Illustration of the Shadow-Jitter attack scenario. Adversarial timing perturbations and log injections distort telemetry observations, inducing misleading correlations in conventional monolithic detection pipelines. In contrast, C-MADF leverages causal filtering and adversarial multi-agent validation to distinguish injected noise from genuine causal signals, thereby reducing false-positive mitigation actions under partial-compromise assumptions.}
    \label{fig:decision_loop}
\end{figure}

\begin{takeawaybox}[C-MADF Design Summary]
The C-MADF architecture integrates causal modeling, constrained decision optimization, adversarial multi-agent deliberation, and calibrated human oversight into a unified framework. Specifically, it: (i) learns a SCM to capture cause–effect relationships in telemetry data; (ii) compiles these constraints into a causally admissible MDP-DAG roadmap that restricts policy search to structurally valid investigation trajectories; (iii) employs an dual-objective multi-agent reinforcement learning with asymmetric reward shaping mechanism in which a Red-Team policy systematically challenges Blue-Team mitigation proposals to reduce false positives and overconfident actions; and (iv) exposes decision-relevant artifacts through a human-in-the-loop interface, where a bounded ETS governs escalation under epistemic uncertainty.
\end{takeawaybox}

\section{Security Analysis}
\label{sec:security_analysis}


In this section, we provide a formal analysis of the security properties of the C-MADF framework. We analyze its resilience against the threats outlined in our threat model and provide arguments for the security guarantees it offers.  {We emphasize that the results here are \emph{verification-aligned} rather than full formal verification in the classical model-checking sense: our theorems give probabilistic bounds on violation frequencies under assumptions about causal mis-specification and adversarial capabilities, but we do not exhaustively explore the state space or prove CTL/LTL properties over the induced MDP-DAG. This distinction follows the model-checking literature, where exhaustive state-space verification is treated as a separate guarantee class} \cite{model_checking}.

\subsection{Resilience to Adversarial Data Manipulation}

A primary threat is adversarial manipulation of telemetry data, such as ``Shadow-Jittering.'' C-MADF's resilience to this threat stems from two components:
\begin{enumerate}
    \item \textbf{Causal Filtering:} The Causal Discovery Module models expected cause--effect relationships. Manipulations that violate learned causal constraints (e.g., high traffic with no corresponding process activity) are flagged as causally inconsistent, filtering unsophisticated perturbations.
    \item \textbf{Adversarial Deliberation:} For causally plausible attacks, the ``Council of Rivals'' provides a second layer of defence. A Red-Team agent challenges the Blue-Team hypothesis by seeking benign explanations for anomalies. The Policy Divergence Score $\mathcal{D}(s_t)$ quantifies ambiguity, discouraging disruptive action under substantial disagreement.
\end{enumerate}

\subsection{Formal Security Properties}

We formalize two key properties of C-MADF:
(i) Causal Consistency of policy execution and
(ii) Adversarial Robustness in terms of false-positive reduction under telemetry manipulation and benign distribution shift.\\


\noindent \textbf{Causal Consistency of Policy Execution}\\

\noindent \textbf{Definition 1 (True Causal DAG).}
Let $\mathcal{G}^\star = (V, E^\star)$ denote the \emph{true} (unknown) causal DAG
governing the security environment, where $V$ is the finite set of investigation states and
$E^\star \subseteq V \times V$ is the set of valid causal transitions.
An edge $(v_i, v_j) \in E^\star$ indicates that $v_j$ is causally admissible only after $v_i$ in the ground-truth system.\\

\noindent \textbf{Definition 2 (Learned Causal DAG and Roadmap Construction).}
Let $\mathcal{G} = (V, E)$ denote the \emph{learned} causal DAG produced by C-MADF.
The corresponding MDP-DAG roadmap is $\mathcal{M} = (S, A, P)$ with $S\subseteq V$, where 
$A$ is a finite action set and $P(s' \mid s, a)$ is a transition kernel satisfying
\[
P(s' \mid s,a) > 0 \Rightarrow (s,s') \in E.
\]
Thus, all realizable transitions in $\mathcal{M}$ are restricted to edges in the learned DAG $\mathcal{G}$.\\

\noindent \textbf{Definition 3 (Spurious-Transition Mass).}
Fix a horizon $T$ and consider trajectories generated by $\mathcal{M}$ under a policy $\pi$.
For each $s\in S$, define the admissible action set
\[
A(s) := \{a\in A : \exists s' \text{ such that } P(s'\mid s,a) > 0\}.
\]
Assume $\pi$ is admissible, i.e., $\pi(a\mid s)=0$ for all $a\notin A(s)$.
Define the per-step spurious-transition probability
\[
\rho_t := \Pr_\pi\!\left[(s_t,s_{t+1}) \in E\setminus E^\star\right],
\]
where $\Pr_\pi(\cdot)$ is over the randomness of $s_0$, the policy, and the transition kernel $P$.
We say C-MADF has bounded spurious-transition mass $\rho$ if $\rho_t \le \rho$ for all $t\in\{0,\dots,T-1\}$.

\begin{theorem}[\textbf{Robust Causal Consistency}]
Let $\mathcal{G}^\star=(V,E^\star)$ be the true causal DAG and $\mathcal{G}=(V,E)$ the learned DAG.
Let $\mathcal{M}=(S,A,P)$ satisfy $P(s' \mid s,a)>0 \Rightarrow (s,s')\in E$.
Let $\tau=(s_0,\dots,s_T)$ be the random trajectory generated under any admissible policy $\pi$.
If $\Pr_\pi[(s_t,s_{t+1}) \in E\setminus E^\star] \le \rho$ for all $t$, then
\[
\mathbb{E}_\pi\!\left[\sum_{t=0}^{T-1}
\mathbf{1}\big[(s_t,s_{t+1})\notin E^\star\big]\right]
\le \rho T.
\]
In particular, if $\rho=0$, then $(s_t,s_{t+1})\in E^\star$ holds almost surely for all $t$.
\end{theorem}


\begin{IEEEproof}
By the roadmap property, $(s_t,s_{t+1})\in E$ holds with probability $1$ for all $t$.
Hence, almost surely,
\[
\mathbf{1}\big[(s_t,s_{t+1})\notin E^\star\big]
=
\mathbf{1}\big[(s_t,s_{t+1})\in E\setminus E^\star\big].
\]

Therefore, by linearity of expectation,
\begin{align*}
\mathbb{E}_\pi\!\left[\sum_{t=0}^{T-1}
\mathbf{1}\big[(s_t,s_{t+1})\notin E^\star\big]\right]
&= \sum_{t=0}^{T-1}
\Pr_\pi\!\left[(s_t,s_{t+1})\in E\setminus E^\star\right] \\
&\le \sum_{t=0}^{T-1}\rho
= \rho T.
\end{align*}
\end{IEEEproof}

 {This result is a bounded-violation guarantee conditioned on the quality of the learned SCM; it does not constitute a global safety proof over all trajectories, nor does it replace exhaustive model checking.}\\

\noindent \textbf{Adversarial Robustness of False-Positive Behaviour}\\

\noindent \textbf{Definition 4 ({Adversarial Perturbation and Induced Measures).}}
Let $(\mathcal{X},\mathcal{F})$ be the telemetry measurable space and $Y\in\{0,1\}$ the ground-truth label.
Let $\mathbb{P}_0$ denote the benign conditional probability measure on $(\mathcal{X},\mathcal{F})$ (the law of $X\mid(Y=0)$), and fix a metric $d$ on $\mathcal{X}$. 
An adversary is a measurable mapping $\mathcal{A}_\delta:\mathcal{X}\to\mathcal{X}$ producing $\tilde{X}=\mathcal{A}_\delta(X)$ such that
$d(\tilde{X},X)\le \delta$ holds $\mathbb{P}_0$-almost surely.
Let $\mathbb{Q}_0$ denote the induced perturbed benign probability measure on $(\mathcal{X},\mathcal{F})$, i.e., the law of $\tilde{X}$ when $X\sim\mathbb{P}_0$.\\

\noindent \textbf{Definition 5 (False-Positive Probability under Perturbed Benign Telemetry).}
Let $(\mathcal{A},\mathcal{G})$ be the action measurable space and let $\mathcal{U}\in\mathcal{G}$ denote disruptive actions.
For any measurable decision rule $f:\mathcal{X}\to\mathcal{A}$, define
\[
P_{\mathrm{FP}}^{\mathbb{Q}_0}(f)
:= \Pr_{\tilde{X}\sim\mathbb{Q}_0}\!\left[f(\tilde{X}) \in \mathcal{U}\right].
\]
\begin{proposition}[\textbf{Monotonic Reduction of False-Positive Mitigations under Gating}]
\label{prop:fp_gating}
Let $f_B:\mathcal{X}\to\mathcal{A}$ be a measurable baseline decision rule and let
$f_C:\mathcal{X}\to\mathcal{A}$ be the gated rule defined by
\[
f_C(x) \in \mathcal{U}
\;\Longleftrightarrow\;
\big(f_B(x) \in \mathcal{U}\big)\wedge G(x),
\]
where $G:\mathcal{X}\to\{0,1\}$ is a measurable gating predicate and $\mathcal{U}\subseteq\mathcal{A}$
denotes the set of disruptive mitigation actions.
Then, under the nominal distribution $\mathbb{Q}_0$,
\[
P_{\mathrm{FP}}^{\mathbb{Q}_0}(f_C) \le P_{\mathrm{FP}}^{\mathbb{Q}_0}(f_B).
\]
Moreover, if there exists $\epsilon>0$ such that
\[
\Pr_{\tilde{X}\sim\mathbb{Q}_0}\!\left[
f_B(\tilde{X}) \in \mathcal{U}\ \wedge\ G(\tilde{X})=0
\right] \ge \epsilon,
\]
then
\[
P_{\mathrm{FP}}^{\mathbb{Q}_0}(f_C) \le P_{\mathrm{FP}}^{\mathbb{Q}_0}(f_B) - \epsilon.
\]
\end{proposition}

\begin{IEEEproof}
By definition,
\[
\{ f_C(\tilde{X}) \in \mathcal{U} \}
=
\{ f_B(\tilde{X}) \in \mathcal{U} \}\cap \{G(\tilde{X})=1\}.
\]
Therefore,
\begin{align*}
P_{\mathrm{FP}}^{\mathbb{Q}_0}(f_C)
&=
\Pr_{\tilde{X}\sim\mathbb{Q}_0}\!\left[f_B(\tilde{X})\in\mathcal{U}\ \wedge\ G(\tilde{X})=1\right] \\
&\le
\Pr_{\tilde{X}\sim\mathbb{Q}_0}\!\left[f_B(\tilde{X})\in\mathcal{U}\right]
=
P_{\mathrm{FP}}^{\mathbb{Q}_0}(f_B).
\end{align*}
For the strict decrease, partition the event $\{f_B(\tilde{X})\in\mathcal{U}\}$ into
$G(\tilde{X})\in\{0,1\}$ to obtain
\begin{align*}
P_{\mathrm{FP}}^{\mathbb{Q}_0}(f_B)
= &\ \Pr[f_B(\tilde{X})\in\mathcal{U},\,G(\tilde{X})=1] \\
  &+ \Pr[f_B(\tilde{X})\in\mathcal{U},\,G(\tilde{X})=0],
\end{align*}
and use the assumption that the second term is at least $\epsilon$.
\end{IEEEproof}

Proposition~\ref{prop:fp_gating} provides a sanity guarantee: gating cannot increase the
rate of disruptive false-positive mitigations under the nominal distribution.
However, this result is monotonicity-based and does not constitute a worst-case
robustness guarantee under arbitrary telemetry manipulation, poisoning, or SCM mis-learning.

\paragraph{Definition 6 (Total Variation Distance).}
The total variation distance between $\mathbb{P}_0$ and $\mathbb{Q}_0$ is
\[
d_{\mathrm{TV}}(\mathbb{P}_0,\mathbb{Q}_0)
:= \sup_{B \in \mathcal{F}} |\mathbb{P}_0(B) - \mathbb{Q}_0(B)|.
\]

\begin{theorem}[\textbf{TV Robustness of False-Positive Rate}]
Let $f:\mathcal{X}\to\mathcal{A}$ be measurable and let $\mathcal{U}\in\mathcal{G}$.
Then
\[
\left|
\Pr_{\tilde{X}\sim\mathbb{Q}_0}[f(\tilde{X}) \in \mathcal{U}]
-
\Pr_{X\sim\mathbb{P}_0}[f(X) \in \mathcal{U}]
\right|
\le
d_{\mathrm{TV}}(\mathbb{P}_0,\mathbb{Q}_0).
\]
\end{theorem}

\begin{corollary}[\textbf{TV-Adjusted FP Bound for C-MADF}]
Under the assumptions above, if
$\Pr_{\tilde{X}\sim\mathbb{Q}_0}[f_B(\tilde{X})\in\mathcal{U}\wedge G(\tilde{X})=0]\ge \epsilon$,
then
\[
P_{\mathrm{FP}}^{\mathbb{Q}_0}(f_C)
\le
P_{\mathrm{FP}}^{\mathbb{P}_0}(f_B)
+
d_{\mathrm{TV}}(\mathbb{P}_0,\mathbb{Q}_0)
-
\epsilon,
\]
where $P_{\mathrm{FP}}^{\mathbb{P}_0}(f_B):=\Pr_{X\sim\mathbb{P}_0}[f_B(X)\in\mathcal{U}]$.
\end{corollary}

\paragraph{Definition 7 (Rejected False-Positive Indicator).}
Let $\tilde{X}\sim\mathbb{Q}_0$. Define
\[
Z := \mathbf{1}\!\left[f_B(\tilde{X}) \in \mathcal{U}\ \wedge\ G(\tilde{X})=0\right],
\qquad
\epsilon := \mathbb{E}_{\tilde{X}\sim\mathbb{Q}_0}[Z].
\]

\begin{theorem}[\textbf{PAC-Style Lower Bound for Rejected-FP Mass}]
Let $\tilde{X}_1,\dots,\tilde{X}_n$ be i.i.d. samples from $\mathbb{Q}_0$ and define
\[
Z_i := \mathbf{1}\!\left[f_B(\tilde{X}_i)\in\mathcal{U}\ \wedge\ G(\tilde{X}_i)=0\right],
\qquad
\hat{\epsilon}_n := \frac{1}{n}\sum_{i=1}^n Z_i.
\]
Then for any $\delta \in (0,1)$, with probability at least $1-\delta$,
\[
\epsilon \ge \hat{\epsilon}_n - \sqrt{\frac{\ln(2/\delta)}{2n}}.
\]
\end{theorem}

\begin{IEEEproof}
Each $Z_i \in [0,1]$ and $\mathbb{E}[Z_i]=\epsilon$.
By Hoeffding's inequality,
\[
\Pr\!\left(|\hat{\epsilon}_n-\epsilon| \ge t\right) \le 2\exp(-2nt^2).
\]
Setting $t=\sqrt{\ln(2/\delta)/(2n)}$ yields the stated bound.
\end{IEEEproof}

\subsection{Verification and Validation}

The structural formulation of C-MADF as a causally constrained MDP-DAG enables formal verification and systematic validation of safety-critical behaviors. 
Because the decision process is restricted to admissible transitions induced by the learned causal graph, the reachable state space is explicitly characterized, making temporal safety properties amenable to algorithmic analysis.\\

\paragraph{Temporal Safety Specifications.}
Let $\mathcal{M}=(S,A,P)$ denote the constrained MDP-DAG and let $\mathcal{D}(s)$ denote the Policy Divergence Score. 
Define the high-uncertainty region
\[
S_{\mathrm{unc}} := \{ s \in S : \mathcal{D}(s) > \tau \},
\]
and let $S_{\mathrm{esc}} \subseteq S$ denote states corresponding to human escalation.

A representative liveness-safety specification can be expressed in Linear Temporal Logic (LTL) as
\[
\mathbf{G}\big(s \in S_{\mathrm{unc}} \rightarrow \mathbf{F}(s \in S_{\mathrm{esc}})\big),
\]
which states that globally, whenever the system enters a high-uncertainty state, it eventually transitions to a human escalation state.\\

\paragraph{Model Checking Perspective.}
Given a finite abstraction of $\mathcal{M}$ (e.g., via discretization of uncertainty levels and bounded horizon), classical model-checking techniques can be applied to verify whether the above LTL property holds over all admissible trajectories. 
In particular, because transitions are constrained by the learned causal DAG, the state-transition graph is acyclic within investigation phases, which reduces verification complexity relative to unconstrained MDPs.\\

\paragraph{Probabilistic Verification.}
When stochastic transitions are retained, probabilistic temporal logics such as PCTL can be used to establish bounds of the form
\[
\Pr\!\left[ \mathbf{F}_{\le T}(s \in S_{\mathrm{esc}}) \mid s_0 \in S_{\mathrm{unc}} \right] \ge 1 - \eta,
\]
ensuring that escalation occurs with high probability within a bounded horizon $T$. 
Such bounds connect directly to the divergence threshold $\tau$ and the gating mechanism defined in Section~\ref{sec:cmadf_framework}.\\

\paragraph{Scope and Limitations.}
We emphasize that these verification procedures rely on the correctness of the learned causal abstraction and finite-state representation. 
They provide guarantees over the induced MDP-DAG, but do not constitute exhaustive verification of the underlying cyber-physical environment.

\section{Experimental Evaluation}
\label{sec:evaluation}

We conduct a comprehensive empirical evaluation on the public CICIoT2023 dataset \cite{ciciot2023} to assess detection robustness, false-positive control, and supervisory-facing explainability of C-MADF.
We compare C-MADF against three recent literature baselines: DUGAT-LSTM, BiLSTM IDS, and Hybrid Weighted-XGBoost, using a shared preprocessing and model-selection protocol. \cite{devendiran2024dugat, imrana2021bidirectional, mohiuddin2023intrusion}

\subsection{Experimental Setup}

\noindent \textbf{Dataset.}
CICIoT2023 is a state-of-the-art, large-scale intrusion-detection benchmark containing extensive benign IoT traffic and multiple modern attack categories (DDoS, DoS, Recon, Web-based, etc.) \cite{ciciot2023}.
We evaluate on this public CICIoT2023 benchmark using its rich network and host telemetry feature set.
Records are preprocessed into the same structured state schema (alert summary, host indicators, and flow statistics), and labels are mapped to a binary attack/benign decision space for C-MADF and all baselines.\\

To reduce evaluation leakage and preserve temporal realism, we use fixed train/validation/test splits at the scenario level, so correlated records from the same scenario family are not distributed across multiple splits. All normalization and feature-scaling statistics are learned on the training split only and then applied unchanged to validation and test data. This prevents optimistic bias from test-informed preprocessing and ensures that all compared methods face identical information constraints.

Because C-MADF is evaluated as a decision-support architecture rather than only a classifier, we map telemetry windows into a structured state representation that preserves temporal context and host--network cross-signals. This representation is shared across all methods through a common interface, so any observed performance difference reflects model behavior rather than differences in feature availability or pre-aggregation privileges.

\noindent \textbf{Baselines.}
To avoid relying on customized internal comparators, we benchmark C-MADF against recent, established literature baselines that are widely used in intrusion detection research.
Specifically, we implement and tune DUGAT-LSTM, BiLSTM IDS, and Hybrid Weighted-XGBoost under the same train/validation/test splits and telemetry channels. \cite{devendiran2024dugat, imrana2021bidirectional, mohiuddin2023intrusion}

For fairness, we apply a shared optimization protocol across all baselines: model-specific hyperparameters are tuned only on the validation split, early-stopping decisions use identical patience rules, and final test metrics are reported only once per seed after model selection. We avoid architecture-specific hand-tuning on the test split and do not apply post-hoc threshold adjustments that are unavailable to competing methods.

\begin{itemize}

\item \textbf{DUGAT-LSTM (Deep IDS Baseline) \cite{devendiran2024dugat}.}
A deep sequence model baseline designed for high-accuracy intrusion detection under complex traffic dynamics. 
For autonomous-response evaluation, its detection outputs are mapped to the same binary attack/benign decision space and action admissibility interface used by C-MADF.

\item \textbf{BiLSTM IDS \cite{imrana2021bidirectional}.}
A strong recurrent baseline for sequential threat-pattern modeling in network telemetry. 
We train and tune this model under identical data partitions and convert prediction outputs to the same intervention decision interface used in our framework.

\item \textbf{Hybrid Weighted-XGBoost \cite{mohiuddin2023intrusion}.}
A competitive ensemble baseline combining meta-heuristic feature weighting with gradient-boosted classification. 
This model is evaluated under the same feature channels and protocol to provide a non-neural cutting-edge reference.

\end{itemize}

\noindent \textbf{Evaluation Metrics.}
Detection performance is evaluated using standard classification metrics:

\begin{itemize}
    \item \textbf{Precision}: proportion of flagged incidents that correspond to genuine attacks.
    \item \textbf{Recall}: proportion of genuine attacks correctly detected.
    \item \textbf{F1-score}: harmonic mean of precision and recall.
    \item \textbf{False Positive Rate (FPR)}: proportion of benign scenarios incorrectly classified as attacks.
\end{itemize}

In addition to detection metrics, we evaluate supervisory-facing interpretability and oversight quality. 
Specifically, we report the ETS and examine its association with evidentiary sufficiency and policy agreement under controlled scenario analysis.

Concretely, ETS is analyzed as an operational gating signal rather than a descriptive score alone: we examine whether lower ETS regions coincide with higher inter-policy disagreement and weaker causal support, and whether higher ETS regions correspond to stable admissible-action consensus. This allows us to assess whether ETS can support practical triage and escalation decisions under uncertain telemetry conditions.

Unless otherwise stated, all reported metrics (precision, recall, F1-score, FPR, ETS) are averaged over three independent random seeds using fixed CICIoT2023 train/validation/test splits.
Table~\ref{tab:results} reports mean values across seeds.
The observed standard deviation across seeds is below 0.02 for all reported metrics.\\

\noindent \textbf{Implementation and Reproducibility.}
All agents are implemented in Python using the PyTorch framework.
For each method and random seed, we fix pseudorandom seeds across data shuffling, neural network initialization, and replay-buffer sampling to ensure experimental consistency.\\

We further log training curves, selected hyperparameters, and per-seed metric summaries for all methods under a unified experiment script. This audit trail enables direct traceability from raw runs to reported tables, and it was used to verify consistency between the main comparison table and the ablation table for the full C-MADF configuration.

\subsection{Performance Comparison on CICIoT2023}

Table~\ref{tab:results} summarizes performance on CICIoT2023. C-MADF achieves a precision of 0.997, recall of 0.961, and F1-score of 0.979, outperforming all three baselines (DUGAT-LSTM: 0.979/0.905/0.941; BiLSTM IDS: 0.982/0.915/0.947; Hybrid Weighted-XGBoost: 0.985/0.941/0.962 for Precision/Recall/F1).
 
Most notably, C-MADF attains a false-positive rate (FPR) of 1.8\%, compared to 11.2\% for DUGAT-LSTM, 9.7\% for BiLSTM IDS, and 8.4\% for Hybrid Weighted-XGBoost. This corresponds to relative FPR reductions of approximately 83.9\%, 81.4\%, and 78.6\%, respectively. Given that false positives in this operational setting may trigger disruptive containment actions affecting grid availability, reducing FPR is particularly important for safety-critical deployment contexts.\\
In practical terms, the gap between 1.8\% FPR (C-MADF) and 8.4\%--11.2\% FPR (baselines) indicates substantially fewer unnecessary interventions under the same evaluation protocol.

Importantly, this reduction is achieved without collapsing recall: C-MADF maintains 0.961 recall, while the strongest baseline recall is 0.941. This pattern argues against a trivial ``always-benign'' explanation and instead supports the intended mechanism of C-MADF---causal admissibility constraints and adversarial policy scrutiny suppress unnecessary actions while preserving sensitivity to genuine attacks.

From an operations perspective, the observed precision--FPR profile is especially relevant for semi-autonomous settings where each false positive may trigger host isolation, traffic blocking, or analyst interruption. Under such conditions, reducing false positives is not only a statistical improvement but also a workflow-stability improvement that can directly reduce defensive self-disruption.
To improve interpretability of these newly added quantitative results, we provide both tabular and visual summaries: Table~\ref{tab:results} gives exact values, while Figure~\ref{fig:results_visual} provides a side-by-side visual comparison across baselines and C-MADF.

\begin{table}[t!]
\centering
\caption{Performance Comparison of C-MADF and Baseline Models {(mean over 3 random seeds; standard deviations $< 0.02$ for all metrics)}}
\label{tab:results}
\begin{tabularx}{\columnwidth}{@{} l c c c c @{}}
\toprule
\textbf{Framework} & \textbf{Precision} & \textbf{Recall} & \textbf{F1-Score} & \textbf{FPR} \\ \midrule
DUGAT-LSTM \cite{devendiran2024dugat} & 0.979 & 0.905 & 0.941 & 11.2\% \\
BiLSTM IDS \cite{imrana2021bidirectional} & 0.982 & 0.915 & 0.947 & 9.7\% \\ 
Hybrid Weighted-XGBoost \cite{mohiuddin2023intrusion} & 0.985 & 0.941 & 0.962 & 8.4\% \\
\textbf{C-MADF} & \textbf{0.997} & \textbf{0.961} & \textbf{0.979} & \textbf{1.8\%} \\ \bottomrule
\end{tabularx}
\end{table}

\begin{figure*}[t]
\centering
\includegraphics[width=\textwidth]{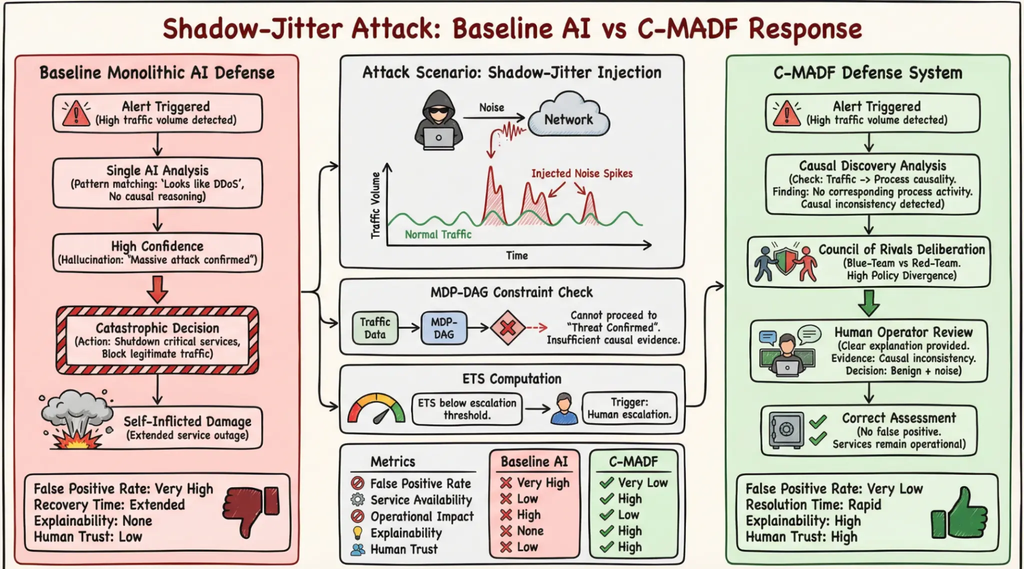}
\caption{Defensive response under Shadow-Jitter injection: C-MADF maintains high detection with low false positives, structured human review, and strong explainability, while the baseline AI fails catastrophically due to hallucinated threat confirmation and lack of verification.}
\label{fig:results_visual}
\end{figure*}

As shown in Figure~\ref{fig:results_visual}, C-MADF's gain is not limited to one metric: it improves precision and F1, preserves competitive recall, and substantially reduces false-positive burden under the shared evaluation protocol. This multi-metric behavior is central to the framework's practical value in high-impact operational settings.

\subsection{Sanity Checks and Baseline Calibration}
\label{subsec:sanity_checks}

To ensure that the observed performance differences are not attributable to artefacts such as distributional bias or inadequate baseline tuning, we conduct a series of validation controls.

\begin{itemize}

\item \textbf{Shared Data Splits and Features.}
All methods (C-MADF, DUGAT-LSTM, BiLSTM IDS, Hybrid Weighted-XGBoost) are trained and evaluated using identical CICIoT2023 splits and the same telemetry feature channels.
No method is provided privileged features or additional telemetry channels.\\

\item \textbf{Hyperparameter Calibration.}
For each baseline, we perform grid search over learning rate, batch size, network width, and exploration parameters using a held-out validation set.
The best-performing configuration per method is selected based on validation F1-score to ensure competitive and properly tuned baselines.\\

\item \textbf{Baseline Scope and Competitiveness Claims.}
DUGAT-LSTM, BiLSTM IDS, and Hybrid Weighted-XGBoost are literature-established cutting-edge baselines selected to test competitiveness against recent external methods under matched data and tuning conditions. Accordingly, the reported superiority claims are grounded in direct comparisons with these contemporary baselines rather than customized internal-only comparators.\\

\item \textbf{Robustness Across Random Seeds.}
All experiments are repeated across three independent random seeds. 
Reported metrics correspond to the mean across seeds. 
Performance differences persist consistently across seeds, with standard deviations below 0.02 for all reported metrics.

\item \textbf{Cross-Table Consistency Control.}
The full-model C-MADF result is required to be numerically identical wherever reused (main comparison and ablation tables), because both derive from the same three-seed evaluation artifact.
This consistency rule was explicitly checked to avoid transcription drift and to make table-to-table interpretation auditable.

\end{itemize}

\subsection{Ablation Study}

To quantify the contribution of individual architectural components, we conduct a structured ablation study in which specific mechanisms of C-MADF are removed while keeping all other training conditions fixed. 
These ablations serve as internal baselines to demonstrate the necessity of each component.

\begin{itemize}

\item \emph{w/o Causal SCM.}
This variant removes the learned Structural Causal Model and relies purely on correlation-based agreement between the Blue and Red teams.

\item \emph{w/o Red Team.}
This variant removes the conservative Red Team, leaving the threat-maximizing Blue Team to operate alone under causal constraints.

\item \emph{w/o Blue Team.}
This variant removes the Blue Team, relying solely on the Red Team and causal filtering.

\end{itemize}

Results are summarized in Table~\ref{tab:ablation}. 
To ensure consistency, both tables use the same three-seed averaging protocol, and the full C-MADF row is identical across both tables (Precision 0.997, Recall 0.961, F1 0.979, FPR 1.8\%). 
Each ablated configuration exhibits measurable degradation relative to the full C-MADF model. 
Removing Causal SCM yields Precision/Recall/F1/FPR of 0.903/0.959/0.930/59.4\%, indicating the importance of filtering causally inconsistent telemetry patterns. 
Disabling Red Team adversarial deliberation yields 0.963/0.963/0.963/21.2\%, suggesting that structured hypothesis testing mitigates reasoning drift under ambiguous evidence. 
Eliminating the Blue Team yields 0.957/0.972/0.964/25.1\%, highlighting the role of the threat-maximization policy in maintaining detection sensitivity.

The full C-MADF model attains the best overall performance, reflecting improved explanatory coherence under integrated causal reasoning, deliberation, and structured decision constraints.

\begin{table}[ht]
\centering
\caption{Ablation Study: Impact of Individual Components on CICIoT2023 {(mean over 3 random seeds; standard deviations $< 0.02$ for all entries)}}
\label{tab:ablation}
\begin{tabularx}{\columnwidth}{@{} l c c c c @{}}
\toprule
\textbf{Configuration} & \textbf{Precision} & \textbf{Recall} & \textbf{F1-Score} & \textbf{FPR} \\ \midrule
\textbf{C-MADF (Full Framework)} & \textbf{0.997} & \textbf{0.961} & \textbf{0.979} & \textbf{1.8\%} \\
\quad w/o Causal SCM & 0.903 & 0.959 & 0.930 & 59.4\% \\
\quad w/o Red Team & 0.963 & 0.963 & 0.963 & 21.2\% \\
\quad w/o Blue Team & 0.957 & 0.972 & 0.964 & 25.1\% \\
\bottomrule
\end{tabularx}
\end{table}

\subsection{ETS Behavior on CICIoT2023}

Beyond classification accuracy, ETS remains consistent with evidentiary sufficiency and policy agreement on CICIoT2023.
High ETS values concentrate on samples with strong causal support and low policy divergence, while lower ETS values are associated with weaker evidence and increased disagreement, supporting calibrated escalation under uncertainty.

This behavior is important because ETS is used at decision time, not merely after-the-fact. In our pipeline, ETS defines a practical confidence envelope: high-ETS decisions are eligible for autonomous execution under policy constraints, mid-ETS decisions are queued for additional evidence collection, and low-ETS decisions are escalated for human adjudication. Framing ETS this way turns explainability from a passive reporting artifact into an active governance mechanism for risk-aware response.

Although this benchmark-grounded analysis cannot substitute for full field validation, the monotonic alignment between ETS and internal reliability indicators provides preliminary evidence that the score is suitable for triage and escalation orchestration in high-tempo SOC workflows.

\begin{takeawaybox}[Empirical Takeaways]
On CICIoT2023, C-MADF records 0.997 precision, 0.961 recall, 0.979 F1-score, and 1.8\% false-positive rate, outperforming all three recent baselines in Table~\ref{tab:results}. This low-FPR result materially reduces unnecessary mitigations and supports safer autonomous response decisions under real-world telemetry conditions.
\end{takeawaybox}

\section{Discussion}
\label{sec:discussion}

The empirical results demonstrate that integrating structural causal constraints, adversarial deliberation, and human-aligned explanation mechanisms can substantially reduce false-positive behavior while preserving detection performance on CICIoT2023. At the same time, evidence from a single benchmark is not sufficient to claim full dependable or verifiable autonomous response under real-world operating conditions.
Rather than interpreting these findings as evidence of deployment readiness, we view them as validation of a design principle: constraining decision policies with causal structure and structured disagreement can mitigate correlation-driven reasoning failures in safety-critical environments.

This section discusses implications, scalability considerations, and limitations of the proposed framework.

\subsection{Implications for Critical Infrastructure Protection}

In safety-critical domains such as power grids, water treatment facilities, and financial clearing systems, false positives can induce operational disruption comparable to the attacks they aim to prevent. 
Our results indicate that structurally constraining policy execution through a learned SCM and MDP-DAG can reduce such self-inflicted disruptions without materially degrading recall.

Importantly, the observed reduction in false-positive rate is not attributable to conservative thresholding alone, as C-MADF remains consistently stronger than all three tuned literature baselines under the same data splits and telemetry channels.
This suggests that structured causal reasoning and hypothesis testing contribute mechanistically to robustness.

The ``Clarity of Command'' principle, maintaining transparent reasoning pathways and calibrated escalation, provides a potential pathway toward certifiable oversight in semi-autonomous defense systems. 
To reduce operator alert fatigue during multi-stage APT campaigns, escalation events can be batched by shared causal root node and attack phase, then prioritized by a composite risk score (policy divergence magnitude, asset criticality, and persistence across time windows) so analysts review a small ranked queue rather than every raw disagreement.
In addition, repeated disagreements linked to the same root-cause hypothesis can be collapsed into a single incident thread with accumulated evidence snapshots, preventing duplicate prompts from overwhelming operators during prolonged campaigns. A configurable ``fatigue budget'' (maximum escalations per analyst per time window) can be enforced by deferring lower-risk disagreements to batched review cycles while immediately surfacing only high-risk, high-persistence events. This policy preserves human attention for decisions with the largest potential operational impact.
However, certification in real infrastructure environments would require validation under live traffic, adversarial red-team testing, and integration with legacy operational workflows.

Because even modern benchmarks like CICIoT2023 are collected in laboratory environments rather than operational deployment logs, external validity remains an open empirical question.

\subsection{Scalability and Generalization}

C-MADF is modular by design: causal discovery, roadmap construction, deliberation, and explanation are separable components. 
This modularity supports adaptation to heterogeneous network environments and alternative telemetry schemas.

However, scaling to enterprise environments with large state spaces introduces challenges. 
Constraint-based causal discovery may become computationally intensive as variable cardinality increases, necessitating distributed or approximate inference techniques. 
In practical IT/OT deployments, the offline SCM can be re-estimated on a weekly cadence for relatively stable OT segments and on a daily or event-triggered cadence for rapidly changing enterprise IT segments, with unscheduled refreshes initiated when drift monitors flag sustained conditional-independence violations.
Operationally, we recommend a two-stage refresh workflow: (i) shadow re-estimation on recent telemetry snapshots, where the candidate SCM is evaluated against current causal-consistency diagnostics, and (ii) controlled promotion to production constraints only if admissibility stability and false-intervention risk metrics remain within predefined tolerance bounds. This reduces the risk of abrupt policy shifts caused by noisy short-term drift.
For very large deployments, the causal layer can be partitioned by subnet or function (e.g., OT control loop, enterprise identity plane, edge IoT segment), with periodic cross-partition reconciliation for shared variables. This hierarchical strategy keeps discovery tractable while preserving enough global structure to prevent contradictory cross-domain interventions.
Similarly, multi-agent deliberation overhead scales with the action hypothesis space, requiring efficient pruning or hierarchical structuring.

Beyond cyber defense, the architectural principle of combining structural causal constraints with adversarial internal review may extend to other high-stakes domains (e.g., medical triage, autonomous systems). 
Nonetheless, such transfer would require domain-specific causal abstractions and careful calibration of escalation thresholds.

\subsection{Limitations and Future Work}

Several limitations merit consideration.

First, the effectiveness of the Causal Discovery Module depends on the fidelity and representativeness of historical telemetry. 
Systematic data poisoning or persistent distribution shift may degrade structural accuracy, affecting downstream policy constraints. 
While our bounded-violation analysis characterizes this dependence formally, empirical robustness under sustained poisoning requires further study.
A practical mitigation path is to combine periodic SCM retraining with drift-triggered challenge sets containing known benign and attack archetypes; significant drops in causal-consistency or intervention-quality metrics would block model promotion and trigger conservative fallback policies.

Second, adversarial deliberation introduces additional inference overhead and assumes sufficiently diverse policy representations between Blue-Team and Red-Team agents. 
If both agents share correlated blind spots, arbitration may fail to surface alternative hypotheses.
Future work should therefore include diversity-promoting training objectives and heterogeneous model families for the two agents, so disagreement reflects substantive epistemic differences rather than noise around a shared bias.

Third, the MDP-DAG roadmap currently relies on a predefined investigation ontology. 
Although this enables verifiable transition constraints, it limits adaptability to unforeseen attack workflows. 
Future work will investigate structure learning for dynamic roadmap construction.
One promising direction is incremental ontology expansion driven by analyst-validated novel incidents, where newly observed investigation motifs are proposed, reviewed, and then compiled into admissible transition templates.

Fourth, ETS was evaluated using benchmark-grounded reliability criteria, evidentiary consistency, and policy stability metrics. While ETS exhibits coherent behavior under these conditions, deployment environments may introduce additional operational constraints and uncertainty sources not captured in benchmark data.
Accordingly, future ETS validation should include human-factors outcomes (decision latency, override frequency, perceived cognitive load) and downstream operational outcomes (avoided false interventions, time-to-containment) to determine whether score calibration transfers from benchmark settings to live SOC practice.

To address external validity concerns directly, future evaluation should extend beyond CICIoT2023 to newer public corpora, production telemetry logs, repeated-seed robustness checks, and operational-impact reporting via avoided false interventions.

Field validation using production logs and structured red-team exercises is necessary to assess real-world generalization and operational integration.

\section{Conclusion}
\label{sec:conclusion}

Autonomous cyber defense systems increasingly operate in environments where correlational reasoning alone is insufficient to ensure safe and reliable decision-making. 
This paper introduced the Causal Multi-Agent Decision Framework (C-MADF), an architecture that integrates SCM, adversarial multi-agent deliberation, and verification-oriented decision constraints to improve robustness and transparency in high-stakes security settings.

Across CICIoT2023 experiments, C-MADF achieved 0.997 precision, 0.961 recall, 0.979 F1-score, and 1.8\% FPR, compared with DUGAT-LSTM (0.979/0.905/0.941, 11.2\% FPR), BiLSTM IDS (0.982/0.915/0.947, 9.7\% FPR), and Hybrid Weighted-XGBoost (0.985/0.941/0.962, 8.4\% FPR).
The ablation study further confirms component necessity: removing Causal SCM, Red Team, or Blue Team degrades performance to 0.903/0.959/0.930/59.4\%, 0.963/0.963/0.963/21.2\%, and 0.957/0.972/0.964/25.1\%, respectively (Precision/Recall/F1/FPR).

In addition, ETS demonstrated alignment with system-internal reliability indicators and evidentiary sufficiency on benchmark telemetry; this should be viewed as preliminary validation of the signal's utility rather than definitive confirmation of real-world trustworthiness. These findings provide empirical support that structurally constrained and deliberative architectures can mitigate correlation-driven reasoning failures in cyber defense.

Future work will focus on large-scale deployment validation, robustness under sustained distribution shift and causal data poisoning, and automated construction of investigation roadmaps. 
Extending evaluation to real-world incident logs and operational red-team exercises will be essential to assess external validity. 
More broadly, the results suggest that integrating causal structure, adversarial reasoning, and human-aligned transparency may offer a promising direction toward dependable semi-autonomous cyber defense systems, contingent on stronger external and human-in-the-loop validation.

\bibliographystyle{IEEEtran}
\bibliography{references}

\end{document}